\documentclass{amsart}
\date{July 23, 2008}

\def\betal{\left(l+\frac12\right)}
 \def\x#1{x_{#1}(l)}
\def\bok{B^{(0)}_\kappa} 
\def\bokb{B^{(0)}_{\kappa^B}}
\def\cok{C^{(0)}_\kappa} 
\def\cokc{C^{(0)}_{\kappa^C}} 

\def\const{\mathrm{const}\;}
\def\tr{\mathop{\mathrm{tr}}\nolimits} % Spur

\def\TF{\mathrm{TF}}
\def\bp{\mathbf{p}}
\def\bq{\mathbf{q}}
\def\bx{\mathbf{x}}
\def\by{\mathbf{y}}

\def\bJ{\mathbf{J}}
\def\bL{\mathbf{L}}

\def\cE{\mathcal{E}}
\def\cO{\mathcal{O}}
\def\cR{\mathcal{R}}
\def\cU{\mathcal{U}} 

\def\gH{\mathfrak{H}}
\def\gQ{\mathfrak{Q}}
\def\gS{\mathfrak{S}}

\def\rd{\mathrm{d}}
\def\ri{\mathrm{i}}
\def\eh{\tfrac12}
\def\H{\mathrm{H}}

\newcommand{\cz}{\mathbb{C}} % Komplexe Zahlen
\newcommand{\zz}{\mathbb{Z}} % Ganze Zahlen
\newcommand{\nz}{\mathbb{N}} % Nat"urliche Zahlen
\newcommand{\rz}{\mathbb{R}} % Relle Zahlen

\newtheorem{theorem}{Theorem}[section]
\newtheorem{proposition}[theorem]{Proposition}
\newtheorem{lemma}[theorem]{Lemma}
\newtheorem{corollary}[theorem]{Corollary}
\theoremstyle{definition}

\title[Scott Correction]{The Energy of Heavy Atoms According to Brown
  and Ravenhall: The Scott Correction}

\author[R. Frank]{Rupert L. Frank} 
\address{Department of Mathematics \\ Princeton
  University\\ Princeton, NJ 08544-1000\\USA}
\email{rlfrank@math.princeton.edu}

\author[H. Siedentop]{Heinz Siedentop}
\address{Mathematisches Institut\\ Ludwig-Maximilians-Universit\"at
  M\"unchen\\ Theresienstra\ss e 39\\ 80333 M\"unchen\\ Germany}
\email{h.s@lmu.de}

\author[S. Warzel]{Simone Warzel}
\address{Department of Mathematics \\ Princeton
  University\\ Princeton, NJ 08544-1000\\USA}
\email{swarzel@princeton.edu}



\begin{document}

\begin{abstract}
  We consider relativistic many-particle operators which -- according
  to Brown and Ravenhall -- describe the electronic states of heavy
  atoms.  Their ground state energy is investigated in the limit of
  large nuclear charge and velocity of light.  We show that the
  leading quasi-classical behavior given by the Thomas-Fermi theory is
  raised by a subleading correction, the Scott correction. Our result
  is valid for the maximal range of coupling constants, including the
  critical one.  As a technical tool, a
  Sobolev-Ga\-gliar\-do-Nirenberg-type inequality is established for
  the critical atomic Brown-Ravenhall operator.  Moreover, we prove
  sharp upper and lower bound on the eigenvalues of the hydrogenic
  Brown-Ravenhall operator up to and including the critical coupling
  constant.%\\[-1.2cm]
\end{abstract}
\maketitle

\tableofcontents
\section{Introduction and main result}
\label{sec:1}
The description of atoms and molecules, in particular of their
energies, has been a primer for the development of quantum mechanics.
However, it became soon clear that atoms with more than one electron
are not accessible to explicit solutions. This motivated the
development of approximate models for large Coulomb systems. One of
the most simple and -- simultaneously -- the most fundamental models
was introduced by Thomas \cite{Thomas1927}, Fermi
\cite{Fermi1927,Fermi1928}, and Lenz \cite{Lenz1932} who proposed the
energy functional which we will also use here. It predicts that the
ground state energy of atoms would decrease with the atomic number $Z$
to leading order as $Z^{7/3}$. In order to get a refined description,
Scott \cite{Scott1952} conjectured that the electrons close to the
nucleus should raise the energy by $Z^2/2$.  Considerably later
Schwinger \cite{Schwinger1980} argued also for Scott's prediction;
Schwinger \cite{Schwinger1981} and Englert and Schwinger
\cite{EnglertSchwinger1984StatisticalAtom:H,EnglertSchwinger1984StatisticalAtom:S,EnglertSchwinger1985A}
even refined these considerations by adding more lower order terms
\cite{Schwinger1981} (see also Englert \cite{Englert1988}).  The
challenge to address the question whether the predicted formulae would
yield asymptotically correct results when compared with the
$N$-particle Schr\"odinger theory was for a long time unsuccessful. It
was Lieb and Simon who proved in their seminal paper
\cite{LiebSimon1977} that the prediction of Thomas, Fermi, and Lenz is
indeed asymptotically correct. However, establishing the Scott
correction resisted the mathematical efforts and became Problem 10B of
Simon's 15 Problems in Mathematical Physics \cite{Simon1984}.
Eventually, the Scott correction was established mathematically by
Hughes \cite{Hughes1986,Hughes1990} (lower bound), and Siedentop and
Weikard
\cite{SiedentopWeikard1987O,SiedentopWeikard1987U,SiedentopWeikard1988,SiedentopWeikard1989,SiedentopWeikard1991}
(lower and upper bound). In fact even the existence of the
$Z^{5/3}$-correction conjectured by Schwinger was proved by Fefferman
and Seco
\cite{FeffermanSeco1989,FeffermanSeco1990,FeffermanSeco1990O,FeffermanSeco1992,FeffermanSeco1993,FeffermanSeco1994,FeffermanSeco1994T,FeffermanSeco1994Th,FeffermanSeco1995}.
Later these results were extended in various ways, e.g., to ions and
molecules.

Despite of the mathematical success in establishing the large $Z$
asymptotics of the Schr\"odinger theory, these considerations remain
questionable from a physical point of view, since large atoms force
electrons into orbits that are close to the nucleus where the
electrons move with high speed which should require a relativistic
treatment. The atom is shrinking with increasing $Z$: already in
non-relativistic quantum mechanincs the bulk of the electrons has a
distance $Z^{-1/3}$ from the nucleus; the electrons contributing to
the Scott correction even live on the scale $Z^{-1}$. Schwinger
\cite{Schwinger1981} has estimated these effects concluding that a
correction to the Scott correction occurs whereas the leading term
should be unaffected by the change of model.  S\o rensen
\cite{Sorensen2005} was the first who proved that the latter is indeed
the case for a simplified ad hoc naive relativistic model, the
Chandrasekhar multi-particle operator, in the limit of large $Z$ and
large velocity of light $c$.  In a previous paper \cite{Franketal2008}
we established the value of the Scott correction which is again of
order $ Z^2 $, a result which was independently announced by Solovej,
S\o rensen, and Spitzer \cite{Solovej2006} (see also S\o rensen
\cite{Sorensen1998} for the non-interacting case). Nevertheless, a
question from the physical point of view remains: Although the
Chandrasekhar model is believed to represent some qualitative features
of relativistic systems, there is no reason to assume that it should
give quantitative correct results.  Therefore, to obtain not only
qualitatively correct results it is interesting, in fact mandatory, to
consider a Hamiltonian which -- as the one by Brown and Ravenhall
\cite{BrownRavenhall1951} -- is derived from QED such that it yields
the leading relativistic effects in a quantitative correct manner.
(See also Sucher \cite{Sucher1980,Sucher1984,Sucher1987}.) The first
step in this direction was taken by Cassanas and Siedentop
\cite{CassanasSiedentop2006} who showed that, similarly to the
Chandrasekhar case, the leading energy is not affected. To show in
which way the Scott correction is changed for this model is our
concern in this paper.

\subsection{Relativistic energy form} 
According to Brown and Ravenhall \cite{BrownRavenhall1951} the energy
of an atom with $N$ electrons in a state $\psi\in \gQ_N^B$ is given by
\begin{equation}\label{eq:1}
   \cE_N^B(\psi) := \left\langle\psi, \left[\sum_{\nu=1}^N \left(
       c\, \boldsymbol\alpha_\nu\cdot \bp_\nu +c^2\beta_\nu -c^2- Z|\bx_\nu|^{-1}\right) + \mkern-8mu \sum_{1\leq\mu<\nu\leq
        N} \mkern-3mu |\bx_\mu-\bx_\nu|^{-1}\right] \psi\right\rangle .
\end{equation}
This involves the free Dirac operator reduced by the rest mass, acting
in $ L^2(\mathbb{R}^3, \mathbb{C}^4) $, with the four Dirac matrices
in standard representation,
$$
\boldsymbol{\alpha}=   \begin{pmatrix} 
    0 & \boldsymbol\sigma\\
    \boldsymbol\sigma & 0
  \end{pmatrix},\quad
  \beta =   \begin{pmatrix}
    1 & 0\\
    0 & -1
  \end{pmatrix},
$$
where $\boldsymbol\sigma$ are the three Pauli matrices in standard
representation, i.e.,
\begin{equation*}
  \sigma_1 =   \begin{pmatrix}
    0& 1\\
    1& 0
  \end{pmatrix}, \quad
  \sigma_2 =  \begin{pmatrix}
    0& -\ri\\
    \ri& 0
  \end{pmatrix}, \quad
  \sigma_3 =  \begin{pmatrix}
    1&0\\
    0&-1
  \end{pmatrix}.
\end{equation*}
We use atomic units in which $m=e^2=\hbar=1$. The parameter $Z$ is the
atomic number and $c$ the velocity of light.  

The Hilbert space of an electron is chosen as the positive spectral
subspace of the Dirac operator,
$$
\gH^B:=\chi_{[c^2,\infty)}(c\, \boldsymbol\alpha\cdot \bp
+c^2\beta)\left(L^2(\rz^3,\cz^4)\right),
$$ 
and, correspondingly, the Hilbert space of $N$ electrons $\gH_N^B$ is
the antisymmetric tensor product of the one-particle space, i.e.,
$\gH_N^B:=\bigwedge_{\nu=1}^N\gH^B$.  Finally, the form domain of
\eqref{eq:1} is $\gQ_N^B:=\gH_N^B \cap \gS(\rz^{3N},\cz^{4^N})$ with
$\gS$ the Schwartz space of rapidly decreasing functions. As is shown
in \cite{Evansetal1996}, the Brown-Ravenhall form $ \cE_N^B $ is
closable and bounded from below if and only if
\begin{equation}
  \label{eq:critical}
  \kappa:= \frac Zc \leq  \kappa^B:=\frac2{2/\pi+\pi/2}.
\end{equation} 
(See also Tix \cite{Tix1997,Tix1998} who improved the bound given in
\cite{Evansetal1996} to an explicite positive bound.) For the physical
value, about $1/137$, of the Sommerfeld fine structure, which equals
$1/c$ in atomic units used here, the critical atomic number $Z$
exceeds $124$ slightly.  This includes all known elements.

In the following we will assume that the atom described by
\eqref{eq:1} is neutral, i.e., $Z=N$, an assumption that we make
mainly for the sake of brevity and clarity of presentation, since the
Scott correction is independent of the ionization degree $N/Z\geq
\const>0$.  Similarly, it might seem that our treatment is restricted
to spherically symmetric systems (atoms).  However, on the energy
scale considered here, molecular Hamiltonians essentially separate --
in nature the distances between nuclei with charges $Z Z_1,..., Z Z_K$
remain on a scale much larger than $Z^{-1/3}$ -- into spherically
symmetric one-center problems (atoms). Therefore, the molecular case
follows from the atomic case by additional localization. However, for
the sake of brevity and clarity, we will spare the reader the
corresponding tedious technicalities, restrict to the atomic case, and
freely use the resulting symmetry.

Thus, according to Friedrichs, the one-particle form $\cE^B_1$
defines for $\kappa\leq \kappa^B$ a distinguished self-adjoint
operator in $ \gH^B $. Through a unitary transformation it may be
represented as a self-adjoint operator in the Hilbert space $\gH :=
L^2(\rz^3,\cz^2)$ of two-spinors.  More precisely, using the notation
$p := |\bp |$, $\boldsymbol{\omega}_\bp := \bp / p$ we set
\begin{align}\label{eq:scalarphi}
  E(p) :=& \sqrt{\bp^2+1},
\qquad \phi_\nu(p):= \sqrt{\frac{E(p) + (-1)^\nu}{2 E(p) }}, \ \nu=0,1,
\end{align} 
and introduce the following bounded operators on $\gH$,
\begin{equation}
  \label{fi1}
  \Phi_0(\bp)  := \phi_0(p) ,\quad
  \Phi_1(\bp):= \phi_1(p) \;\, \boldsymbol\sigma\cdot \boldsymbol{\omega}_\bp .
\end{equation}
The operator $ \boldsymbol{\Phi}_c : \, \gH \to \gH^B $, $ \psi
\mapsto \left( \Phi_0(\bp/c) \,\psi , \Phi_1(\bp/c)\, \psi \right) $,
embeds $\gH $ unitarily into $\gH^B$ \cite{CassanasSiedentop2006}.
Therefore, the form $ \cE^B_1 $ defines the (two-spinor)
Brown-Ravenhall operator in $ \gH $,
\begin{equation}\label{eq:BR}
  B_c[Z/|\bx|]:=\boldsymbol{\Phi}_c^{-1} \left(c
    \boldsymbol\alpha\cdot \bp +c^2\beta -c^2 - Z/|\bx| \right)
  \boldsymbol{\Phi}_c = c^2 E(p/c) - c^2 - \cU_c(Z/|\bx|) ,
\end{equation}
where $ \cU_c(A) := \Phi_0(\bp/c)\, A \,\Phi_0(\bp/c) + \Phi_1(\bp/c)
\, A \, \Phi_1(\bp/c) $.  In the case $c=1$ we denote this operator by
$B_Z$.  Further properties properties of $B_Z$ and its relation to the
corresponding Chandrasekhar operator and Schr\"odinger operator
\begin{equation}
  \label{CS}
  C_Z:=(\bp^2+1)^{1/2}-1-Z/|\bx|,\qquad S_Z:=\tfrac12\bp^2-Z/|\bx|
\end{equation}
all realized in $\gH$, can be found in Sections~\ref{sec:crit} and
\ref{sec:BC} below and in Appendix \ref{wasserstoff}.

%%%%%%%%%%%%%%%%%%%%%%%%%%%%%%%%%%%%%%%%%%%%%%%%%%%%%%%%%%%%%%%%%%%%%5

\subsection{Main result} 
We are interested in the ground state energy
\begin{equation*}
  E^B_c(Z) := \inf\{\cE_Z^B(\psi)\, |\, \psi\in\gQ_Z^B,\ \|\psi\|=1\} 
\end{equation*} 
of the energy form \eqref{eq:1} for large atomic number $Z$ and large
velocity of light $ c $ satisfying \eqref{eq:critical}. Note that we
picked $N=Z$. It was shown in \cite{CassanasSiedentop2006}, that
similarly to the Chandrasekhar case \cite{Sorensen2005}, the leading
behavior of $E^B_c(Z)$ is not affected by relativistic effects and, as
in the Schr\"odinger case \cite{LiebSimon1977}, given by the minimal
Thomas-Fermi energy
\begin{equation}
  \label{eq:minimum}
  E_\mathrm{TF}(Z):= \inf\{\cE_\mathrm{TF}(\rho)\, | \, \rho\in
  L^{5/3}(\rz^3),\ \rho\geq0,\ D(\rho,\rho)<\infty\} .
\end{equation}
The latter is defined in terms of the 
Thomas-Fermi energy functional 
\begin{equation*}
  \cE_\mathrm{TF}(\rho)
  :=  \int_{\rz^3}\left[\frac{3}{5}\gamma_\mathrm{TF}\,\rho(\bx)^{5/3} - \frac
    Z{|\bx|} \rho(\bx)\right]\rd \bx + D(\rho,\rho)
\end{equation*}
where, in our units, $\gamma_\mathrm{TF}=(3\pi^2)^{2/3}/2$ and
\begin{equation*}
  D(\rho,\sigma):=\frac12\int_{\rz^3 }
  \int_{\rz^3}\frac{\overline{\rho(\bx)}\sigma(\by)}{|\bx-\by|} \,\rd \bx \rd\by
\end{equation*}
is the Coulomb scalar product.  By scaling, one finds
$E_\mathrm{TF}(Z) = E_\mathrm{TF}(1) \, Z^{7/3}$.

This paper concerns the correction to the leading behavior. For the
formulation of the main result, we abbreviate the negative part of an
operator by $A_- := - A\chi_{(-\infty,0)}(A)$ and introduce for
$0<\kappa\leq\kappa^B$ the spectral shift
\begin{equation}
  \label{eq:scott}
  s(\kappa) := \kappa^{-2} \tr_{\gH}\left[\left({B_\kappa}\right)_- -\left( {S_\kappa}\right)_- \right] .
\end{equation}
(We use the term ``spectral shift'' for $s$ for convenience although
it is used in slightly different meaning otherwise.)  It describes the
shift of the Brown-Ravenhall bound state energies compared to those of
the Schr\"odinger operator. In Section~\ref{sec:BC} we show that $s$
is well-defined and discuss some of its properties. In particular, we
prove that the function $s$ is continuous and non-negative on the
interval $(0,\kappa^B]$ and satisfies
\begin{equation}\label{eq:szero}
	s(\kappa)=\cO(\kappa^2)
\qquad\text{as}\ \kappa\to 0.
\end{equation}

We are now ready to state our main result.
\begin{theorem}[\textbf{Scott correction}]
  \label{t2}
  There exists a constant $C>0$ such that for all $c\geq Z/\kappa^B$
  and for all $Z\geq 1$ one has
\begin{equation}\label{eq:main}
  \left|E_{c}^B(Z)-E_\mathrm{TF}(Z)-\left(\tfrac12-s(Z/c)\right)Z^2\right| 
  \leq C Z^{47/24}.
\end{equation}
\end{theorem}

Put differently, Theorem~\ref{t2} asserts that in the limit
$Z\to\infty$ we have uniformly in the quotient $\kappa=Z/c \in
(0,\kappa^B]$
\begin{equation}
  \label{th:eqt1}
  E_c^B(Z) =E_\mathrm{TF}(Z) + \left(\tfrac12 - s(\kappa)\right) Z^2 + o(Z^2) .
\end{equation}
(We do not claim that the error $Z^{47/24}$ in \eqref{eq:main} is
sharp, so we only write $o(Z^2)$ here.) The second term
$\left(\tfrac12 - s(\kappa)\right)Z^2$ in \eqref{th:eqt1} is the
so-called Scott correction in the Brown-Ravenhall model.  It does not
exceed the Scott correction $Z^2/2$ in the non-relativistic model
\cite{SiedentopWeikard1987O}. Indeed, if $\kappa=Z/c$ stays away from
zero then there is a relativistic lowering of the ground state energy
at order $Z^2$. On the other hand, in the non-relativistic limit
$c\to\infty$ with $ \kappa=Z/c \to 0 $, one recovers --
non-surprisingly -- the value of the Schr\"odinger case. In this case
\eqref{eq:szero} implies
\begin{equation}
  \label{explicit}
  E_{c}^B(Z) = E_\mathrm{TF}(Z) + \tfrac12Z^2 + \cO(c^{-2}Z^{4} + Z^{47/24}).
\end{equation}
The Scott correction in the Brown-Ravenhall model, however, exceeds
the Scott correction predicted by the naive Chandrasekhar model
treated in \cite{Franketal2008} and announced in \cite{Solovej2006}.
This follows from the fact that sums of bound state energies of the
atomic Chandrasekhar operator are dominated by those of the
Brown-Ravenhall operator, cf. the proof of Theorem~\ref{t:3} below.

%%%%%%%%%%%%%%%%%%%%%%%%%%%%%%%%%%%%%%%%%%%%%%%%%%%%%%%%%%%%%%%%%%%%%%%%

\subsection{Outline of the paper}\label{sec:outline}

The central strategy of our paper is to compare the ground state
energy of the Brown-Ravenhall operator with that of the Schr\"odinger
operator. The latter is known up to the required accuracy $o(Z^2)$ and
the leading contribution agrees with the Brown-Ravenhall energy. The
subtraction of the corresponding ground state energies results in a
renormalized effective model which accurately describes the energy
differences and is amenable to analysis. The germ of this idea has
been presented in the simpler context of the Chandrasekhar
model~\cite{Franketal2008}.  The full blown renormalization required
is developed in this paper. A virtue of our approach is that it leads
to an explicit formula for the spectral shift which can be evaluated
numerically. We believe it would be interesting to compare this
formula with experimental data.

We show that the difference between the Brown-Ravenhall and
Schr\"odinger ground state energies on the multi-particle level
coincides, up to the required accuracy, with a spectral shift on the
one-particle level. A crucial step in our analysis is therefore a
bound on the corresponding spectral shift for rather general
spherically symmetric potentials.  This is presented in Section
\ref{sec:BC}, where we show that sums of differences of
Brown-Ravenhall and Schr\"odinger eigenvalues decay rather rapidly as
the angular momentum increases.

In Section \ref{sec:crit} we address various aspects of hydrogenic
Brown-Ra\-ven\-hall operators. An essential feature and source of
difficulties, which does not occur in the naive Chandrasekhar model,
is the non-locality of the potential energy. In particular, instead of
the usual Coulomb potential $ |\bx|^{-1} $ we face the `twisted'
non-local operator $ \cU_c( |\bx|^{-1} ) $. Estimating the difference
between the corresponding potential energies is the topic in
Subsection~\ref{sec:twisted}. Since, in contrast to the Schr\"odinger
case, the eigenvalues of the hydrogenic Brown-Ravenhall operator are
not known explicitly, we prove upper and lower bounds in Subsection
\ref{wasserstoff2}. Our bounds are sharp with respect to their
dependence on the quantum numbers $n$ and $l$. An upper bound is given
by the Dirac eigenvalues, a consequence of the mini-max principle for
eigenvalues in the gap. For the lower bound we overcome the
non-locality of the potential by a non-trivial comparison argument
with a super-critical Chandrasekhar operator. In Subsection
\ref{sec:sobolev} we prove a new Sobolev-type inequality, from which
we derive estimates on the eigenfunctions of the hydrogenic
Brown-Ravenhall operator. The technical challenge here is to prove
such a result up to and including the critical coupling constant.

Finally, we present the proof of our main result, Theorem \ref{t2}, in
Section~\ref{sec:2}.

For the readers' convenience we collect various facts in the
appendices. Appendix~\ref{app:J} recalls the partial wave
decomposition of the Hilbert space of two-spinors, Appendix
\ref{twisting} establishes some useful properties of the twisting
operators, and Appendix \ref{wasserstoff} collects basic facts on
hydrogenic Brown-Ravenhall and Chandrasekhar operators. Appendix
\ref{sec:liebyau} fills in some details in the proof of Theorem
\ref{thm:evhydc} and, eventually, Appendix \ref{app:a} defines the
one-particle density matrix giving the main contribution of the
energy.

\section{The hydrogenic Brown-Ravenhall operator\label{sec:crit}}

In this section we set $ c = 1 $ and investigate the Brown-Ravenhall
operator with Cou\-lomb potential 
\begin{equation}
  \label{eq:12a}
  B_\kappa =\sqrt{\bp^2 - 1} - 1 - \kappa\cU(|\bx|^{-1})
\end{equation}
in the Hilbert space $\gH=L^2(\mathbb{R}^3,\mathbb{C}^2)$ of
two-spinors, where we recall that
\begin{equation}
  \label{def:Ucal}
  \cU(|\bx|^{-1})=\Phi_0(\bp)|\bx|^{-1}\Phi_0(\bp)+\Phi_1(\bp)|\bx|^{-1}\Phi_1(\bp)
\end{equation}
with $\Phi_\nu$ defined in \eqref{fi1}.  In Subsection
\ref{wasserstoff2} we prove sharp upper and lower bounds on the
eigenvalues of $ B_\kappa $. In Subsection \ref{sec:sobolev} we prove
$L^p$ estimates on the eigenfunctions of this operator.  Technically,
this is expressed as a Sobolev-type inequality for the massless
version of $B_\kappa$, which is a non-negative operator. Finally, in
Subsection \ref{sec:twisted} we compare the potential energy of the
operator $B_\kappa$, namely $\left\langle
  \psi,\cU(|\bx|^{-1})\psi\right\rangle$, with the corresponding local
potential energy $\left\langle \psi,|\bx|^{-1} \psi\right\rangle$. For
comparison purpose also the corresponding Chandrasekhar and
Schr\"odinger operator $C_\kappa$ and $S_\kappa$ occur (see
\eqref{CS}).

According to \cite{Evansetal1996} and \cite{Kato1966}
the operators $ B_\kappa $ and $ C_\kappa $ are well-defined
for all $ \kappa \leq \kappa^\# $ with $ \# = B, C $ and 
\begin{equation}\label{eq:couplingcrit}
\kappa^B = \frac{2}{2/\pi+\pi/2},
\qquad
\kappa^C := 2/\pi ;
\end{equation}
see also Appendix~\ref{wasserstoff}.
Of course, for the
Schr\"o\-din\-ger operator no upper bound on $\kappa$ is needed. 
%%%%%%%%%%%%%%%%%%%%%%%%%%%%%%%%%%%%%%%%%%%%%%%%%%%%%%%%%%%%%%%

\subsection{Estimates on eigenvalues of the hydrogen atom\label{wasserstoff2}}
In contrast to the Schr\"odinger or Dirac models, the eigenvalues of
$B_\kappa$ and $C_\kappa$ are not known explicitly. In order to obtain
upper and lower bounds on these eigenvalues, we use that the spectra
of $B_\kappa$, $ C_\kappa $ and $ S_\kappa$ may be classified in terms
of angular momenta.

As usual write $\bL:=\bx\times \bp$ for the operators of orbital angular
momentum and $\bJ:= \bL +\eh \boldsymbol\sigma$ for the operators of
total angular momentum.  The four operators $B_\kappa$, $\bJ^2$, $J_3$,
$\bL^2$ commute pairwise, and this also holds, if $C_\kappa$ or $S_\kappa$
replace $B_\kappa$.  This allows a decomposition of the Hilbert space
$\gH$ into orthogonal subspaces which reduce such a quadruple of
operators, i.e.,
\begin{equation}\label{eq:decomp}
  \gH = \bigoplus_{j \in \nz_0 + \frac{1}{2} } \bigoplus_{l=j \pm 1/2} \gH_{j,l} ,
  \qquad \gH_{j,l}:=\bigoplus_{m=-j}^j \gH_{j,l,m} .
\end{equation}
Here $\gH_{j,l,m}$ is the maximal joint eigenspace of $\bJ^2$ with
eigenvalues $j(j+1)$, of $\bL^2 $ with eigenvalue $l(l+1)$, and $J_3$
with eigenvalue $m$. More details concerning the partial wave
decomposition~\eqref{eq:decomp} can be found in Appendix \ref{app:J}.

We denote by $b_{j,l}(\kappa)$, $c_l(\kappa)$, and $s_l(\kappa)$ the
reduced operators corresponding to fixed angular momenta $ j $ and $ l
$, where, strictly speaking, we consider $ b_{j,l}(\kappa) $ and $
c_l(\kappa) $ in momentum space whereas $s_l(\kappa) $ in position
space. We refer to Appendix \ref{wasserstoff} for precise definitions
and further discussion.

The main result of this subsection is that for large quantum numbers
$n$, $j$, and $ l $, the eigenvalues of $ b_{j,l}(\kappa) $ and $
c_l(\kappa) $ behave similarly to the explicitly known ones of the
Schr\"odinger operator $ s_l(\kappa) $.

\begin{theorem}[\textbf{Energies of Brown-Ravenhall hydrogen}\label{thm:evhydbr}]
  There is a constant $C< \infty $ such that for all
  $j\in\nz_0+\tfrac{1}{2}$, and $l = j \pm \tfrac{1}{2}$, $n\in\nz$
  and $\kappa\in (0,\kappa^B]$ one has
  \begin{equation}\label{eq:evhydbr}
    -C\frac{\kappa^2}{(n+l)^{2}}\leq \lambda_n(b_{j,l}(\kappa)) \leq -\frac{\kappa^2}{2 (n+l)^{2}}.
  \end{equation}
\end{theorem}
Here and below, we denote by $\lambda_1(A)\leq \lambda_2(A)\leq
\ldots$ the eigenvalues, repeated according to multiplicities, below
the bottom of the essential spectrum of the self-adjoint, lower
semi-bounded operator $A$.  Note that $-\kappa^2(2 (n+l)^{2})^{-1} =
\lambda_n(s_l(\kappa))$ on the left hand side of \eqref{eq:evhydbr} is
the $n$-th eigenvalue of the Schr\"odinger operator corresponding to
angular momentum $l$. In particular, we conclude from
\eqref{eq:evhydbr} that for all $\mu\geq 0$
  \begin{equation}\label{eq:t3a}
    0\leq \tr_{j,l} \left( \left[B_\kappa+\mu \right]_- 
      - \left[ S_\kappa + \mu \right]_-\right) < \infty.
  \end{equation}

  In the proof of Theorem \ref{thm:evhydbr} we use heavily the
  corresponding result for the Chandrasekhar case, which we state
  next.
\begin{theorem}[\textbf{Energies of Chandrasekhar hydrogen}\label{thm:evhydc}]
  There is constant $C< \infty$ such that for all $l\in\nz_0$, $n\in\nz$ and
  $\kappa\in (0,\kappa^C]$ one has
  \begin{equation}\label{eq:evhydc}
    -C \frac{\kappa^2}{(n+l)^{2}}\leq \lambda_n(c_{l}(\kappa)) 
    \leq -\frac{\kappa^2}{2(n+l)^{2}}.
  \end{equation}
\end{theorem}

We break the proofs of Theorems \ref{thm:evhydbr} and \ref{thm:evhydc}
into three parts, corresponding to the upper bound and the lower bound
for subcritical and, respectively, critical values of the coupling
constant.

%%%%%%%%%%%%%%%%%%%%%%%%%%%%%%%%%%%%%%%%%%%%%%%%%%%%%%%%%%%%%%%%%%%%%%%%%%%%%%%

\subsubsection{Upper bound on hydrogen eigenvalues}

We begin with the Chandrasekhar case.
\begin{proof}[Proof of Theorem \ref{thm:evhydc}. Upper bound.]
  The second inequality in \eqref{eq:evhydc} is an immediate
  consequence of the inequality $\sqrt{p^2+1}-1 \leq p^2/2$ and the
  known form of the Schr\"odinger eigenvalues in the subspace
  corresponding to fixed angular momentum~$l$.
\end{proof}

Next, we turn to the Brown-Ravenhall case.
\begin{proof}[Proof of Theorem \ref{thm:evhydbr}. Upper bound.]
  We first recall some facts about the eigenvalues of the
  hydrogenic  Dirac operator
$D_\kappa:= \boldsymbol\alpha\cdot \bp +\beta - \kappa |\bx|^{-1} $; see Darwin
  \cite{Darwin1928}, Gordon \cite{Gordon1928} and also Bethe and
  Salpeter \cite{BetheSalpeter1957} for a textbook presentation. The
  following subspaces of $L^2(\rz^3,\cz^4)$,
  $$
  \tilde\gH_{j,l,m} = \left\{ \bx \mapsto \left(\begin{matrix} i r^{-1} f(r)
        \Omega_{j,l,m}(\omega_\bx) \\ - r^{-1} g(r) \Omega_{j,2j-l,m}(\omega_\bx)
      \end{matrix}\right) :\ f,g\in L^2(\rz_+) \right\},
  $$
  reduce the Dirac operator $D_\kappa$ with $ \kappa \in (0,1) $.  Under the
  natural identification of $\tilde\gH_{j,l,m}$ with
  $L^2(\rz_+,\cz^2)$ the part of $D(\kappa )$ in $\tilde\gH_{j,l,m}$
  is unitarily equivalent to
 $$
  d_{j,l}(\kappa) =
  \begin{pmatrix} 
    1-\tfrac\kappa r & -\tfrac{\rd}{\rd r} - \tfrac{(j-l)(2j+1)}{r} \\
    \tfrac{d}{dr} - \tfrac{(j-l)(2j+1)}{r} & -1-\tfrac\kappa r
  \end{pmatrix}.
  $$
    The non-decreasing sequence $\lambda_n(d_{j,l}(\kappa))$ of
  eigenvalues of $d_{j,l}(\kappa)$ in the gap $(-1,1)$ is independent
  of $l$ and given explicitly by
  \begin{equation}
	\label{dirac}
  \lambda_n(d_{j,l}(\kappa)) = \left(1-
    \frac{\kappa^2}{\left(n-1+\sqrt{(j+1/2)^2-\kappa^2}\right)^2 + \kappa^2}
  \right)^{1/2}, \quad n\in\nz.
  \end{equation}
  The Dirac eigenvalues reduced by the rest energy are bounded from
  above by the Schr\"o\-ding\-er eigenvalues: for all $n$, $ l $, $j$,
  and $ \kappa \in (0,1) $
  \begin{equation}\label{d<s}
  1-\lambda_n(d_{j,l}(\kappa)) \geq \tfrac{\kappa^2}{2(n+l)^2} =
  -\lambda_n(s_l(\kappa)).
  \end{equation}
  To show \eqref{d<s}, we use $ \sqrt{(j+1/2)^2-\kappa^2}
  \leq \sqrt{(l+1)^2-\kappa^2} \leq\sqrt{(n+l)^2 - \kappa^2} + 1 - n$
  and expand the outer square root in \eqref{dirac} up to first order which gives an
  upper bound.

  Hence the assertion will follow, if we can show that
  \begin{equation}\label{eq:brdirac}
  \lambda_n(b_{j,l}(\kappa)) \leq -1+\lambda_n(d_{j,l}(\kappa)).
  \end{equation}
  To prove this, we fix $(j,l)$ and abbreviate $\Lambda_+ :=
  \chi_{[1,\infty)}(d_{j,l}(0))$ and $\Lambda_-:=1-\Lambda_+$. It
  follows from the definition of the Brown-Ravenhall operator that
  $b_{j,l}(\kappa)$ is unitarily equivalent to the operator
  $\Lambda_+(d_{j,l}-1)\Lambda_+$ in the Hilbert space
  $\Lambda_+(L^2(\rz_+,\cz^2))$. The variational principle for
  eigenvalues in gaps by Griesemer et al.
  \cite{Griesemeretal1999,GriesemerSiedentop1999} under the weakened
  hypotheses of Dolbeault et al. \cite{Dolbeaultetal2000O} states that
  \begin{equation*}
    \lambda_n(d_{j,l}(\kappa)) 
    = \inf_{V\subset\Lambda_+(L^2(\rz_+,\cz^2)),\atop \dim V=n}
    \sup\left\{ \frac{(f,d_{j,l}(\kappa) f)}{\|f\|^2} :\ 0\neq f\in V\oplus 
      \Lambda_-(L^2(\rz_+,\cz^2)) \right\}.
  \end{equation*}
  Since the supremum decreases when restricted to $0\neq f\in V$, one
  obtains~\eqref{eq:brdirac}.
\end{proof}

%%%%%%%%%%%%%%%%%%%%%%%%%%%%%%%%%%%%%%%%%%%%%%%%%%%%%%%%%%%%%%%%%%%%%%%%%%%

\subsubsection{Lower bounds on hydrogen eigenvalues. Subcritical case}

\begin{proof}[Proof of Theorem~\ref{thm:evhydc}. Subcritical case]

  Since we will reduce the Brown-Ravenhall case in
  Theorem~\ref{thm:evhydbr} to the Chandrasekhar case, we actually
  prove a slightly stronger statement. As explained in
  \eqref{rem:critconst}, the operators $ c_l(\kappa) $ are lower
  bounded for all $ l \geq 1 $ up to $ \kappa^C_l > \kappa^B $.

  We assume that either $l\geq 1$ and $0<\kappa \leq \kappa^B$ or else
  that $l=0$ and $0<\kappa\leq \kappa^B \kappa^C/\kappa^C_1$. For any
  $0<\delta<1$ there exist $M_\delta>0$ and $c_\delta>0$ such that 
$$ 
\sqrt{p^2+1} - 1 \geq 
\begin{cases}
	(1-\delta)p & \mbox{if $ p\geq M_\delta$} \\
    c_\delta\, p^2/ 2 & \mbox{if $ p\leq M_\delta$.}
\end{cases}.
$$
Denoting by $\chi_i$ the characteristic function of the centered ball
in $\rz^3$ with radius $M_\delta$, and putting $\chi_o:=1-\chi_i$, the
Schwarz inequality implies the operator inequality
$$
|\bx|^{-1} \leq (1+\delta^{-1}) \chi_i(\bp) |\bx|^{-1} \chi_i(\bp) +
(1+\delta) \chi_o(\bp) |\bx|^{-1} \chi_o(\bp) ,
$$
and hence
\begin{align}\label{eq:momentumcut}
  \sqrt{\bp^2+1}-1-\kappa |\bx|^{-1}
  \geq & \chi_i(\bp) \left(c_\delta\bp^2/2 - (1+\delta^{-1})\kappa |\bx|^{-1}\right) \chi_i(\bp) \\
  & + \chi_o(\bp) \left( (1-\delta) |\bp| - (1+\delta)\kappa
    |\bx|^{-1}\right) \chi_o(\bp) \,. \notag
\end{align}
Now choose $\delta$ as the the unique solution of the equation
$(1+\delta)/(1-\delta)=\kappa^C_1/\kappa^B$ in the interval $(0,1)$.
Then the restrictions on $\kappa$ imply that $(1+\delta)\kappa \leq
(1-\delta)\kappa^C_1\leq(1-\delta)\kappa^C_l$ for $l\geq 1$ and
$(1+\delta)\kappa \leq (1-\delta)\kappa^C$ for $l=0$. In any case, the
second operator in the above sum is non-negative. The variational
principle hence implies that the $n$-th eigenvalue of $c_l(\kappa)$ is
greater or equal to the $n$-th eigenvalue of $\chi_i(\bp)
\left(c_\delta\bp^2/2 - (1+\delta^{-1})\kappa |\bx|^{-1}\right)
\chi_i(\bp)$.  Again by the variational principle, the latter is
greater or equal to the $n$-th eigenvalue of $c_\delta\bp^2/2 -
(1+\delta^{-1})\kappa |\bx|^{-1}$, which is $-\const \kappa^2
(n+l)^{-2}$.
\end{proof}

\begin{proof}[Proof of Theorem~\ref{thm:evhydbr}. Subcritical case]
  We assume that either $j\geq 3/2$ and $0<\kappa \leq \kappa^B$ or
  else that $j=1/2$ and $0<\kappa\leq \kappa^B \kappa^C/\kappa^C_1$.
  We claim that
  \begin{equation}\label{eq:evbrc}
    \lambda_n( c_l(\kappa)) = \lambda_{2n-1}( c_l(\kappa)\otimes {1}_{\cz^2})
    \leq \lambda_{2n-1}( b_{j,l}(\kappa)) \,.
  \end{equation}
  Once we have proved this, the assertion follows easily from what we
  have shown in the proof of Theorem~\ref{thm:evhydc} above.

  To establish \eqref{eq:evbrc} we use the same notation as in the
  proof of the upper bound in Theorem~\ref{thm:evhydbr}. By the
  variational principle,
\begin{align*}
  & \lambda_n( b_{j,l}(\kappa))\\
  = & \sup_{\substack{f_1,\ldots,f_{n-1}\in \\ \Lambda_+(L^2(\rz_+,\cz^2))} } \!\!\!
  \inf\{ \langle f, (d_{j,l}(\kappa) -1)f\rangle \big| \ \|f\|=1, f\in \Lambda_+(L^2(\rz_+,\cz^2)), f \perp f_\nu \}\\
  = & \sup_{\substack{f_1,\ldots,f_{n-1}\in \\ L^2(\rz_+,\cz^2)} } \!\!\!
  \inf\{ \langle \mathcal F_l f, c_l(\kappa) \mathcal F_l f \rangle: \
  \|f\|=1, f \in\Lambda_+(L^2(\rz_+,\cz^2)), f \perp f_\nu \}
  \end{align*}
  with $\mathcal F_l$ the Fourier-Bessel transform, see
  \eqref{eq:FourierBessel}. The infimum does not increase if the
  condition $f \in\Lambda_+(L^2(\rz_+,\cz^2))$ is relaxed to $f \in
  L^2(\rz_+,\cz^2)$. This gives the eigenvalues of the operator
  $c_l(\kappa) \otimes {1}_{\cz^2}$, proving \eqref{eq:evbrc}.
\end{proof}

%%%%%%%%%%%%%%%%%%%%%%%%%%%%%%%%%%%%%%%%%%%%%%%%%%%%%%%%%%%%%%%%%%%%%%%%%%%%%%

\subsubsection{Lower bounds on hydrogen eigenvalues. Critical case}

\begin{proof}[Proof of Theorem~\ref{thm:evhydc}. Critical case]
  It remains to prove that
  $$
  \lambda_n(c_0(\kappa)) \geq - \const \kappa^2 n^{-2}
  $$
  for $\kappa^B \kappa^C/\kappa^C_1 \leq \kappa \leq \kappa^C$. We may
  assume that $\kappa = \kappa^C$ and will prove that for all $\tau>0$
  \begin{equation}\label{eq:l=0zaehlen}
    N(-\tau, c_0(\kappa^C)) 
    := \tr \chi_{(-\infty,-\tau)}(c_0(\kappa^C)) \leq \const\tau^{-1/2} .
  \end{equation}
  Let $\chi_i^2+\chi_o^2=1$ be a smooth radial quadratic partition of
  unity with $\chi_i$ supported in the unit ball and $\chi_o$
  supported outside the ball of radius $1/2$ about the origin. It was
  shown in \cite[Eq.~(19)]{Franketal2008} that the localization error
  can be estimated by a bounded exponentially decaying potential
  $v(r)\leq\const e^{-r}$, i.e.,
  \begin{multline*}
    \sqrt{p^2+1}-1-\kappa^C |\bx|^{-1}
    \geq \chi_i\left( \sqrt{p^2+1}-1-\kappa^C |\bx|^{-1} - v(|\bx|)\right) \chi_i\\
    + \chi_o\left( \sqrt{p^2+1}-1-\kappa^C |\bx|^{-1} -
      v(|\bx|)\right) \chi_o .
  \end{multline*} 
  By the variational principle it suffices to consider the eigenvalue
  counting function corresponding to the interior and exterior term
  separately.  The interior term is further estimated according to
  $$
  \chi_i\left( \sqrt{p^2+1}-1-\kappa^C |\bx|^{-1} - v(|\bx|)\right)
  \chi_i \geq \chi_i\left( |\bp|-\kappa^C |\bx|^{-1} - \const\right)
  \chi_i \,.
  $$ 
  As shown by Lieb and Yau \cite{LiebYau1988} and explained in
  Corollary~\ref{lem:LiebYau}, the number of negative eigenvalues of
  the latter operator acting in the subspace corresponding to $l=0 $
  is finite, i.e., for all $ \tau > 0 $
  \begin{equation}\label{eq:LiebYaunumber}
    N_{l=0}\left(-\tau, \chi_i\left( |\bp|-\kappa^C |\bx|^{-1} - \const\right) \chi_i \right) \leq \const .
  \end{equation}
  For the exterior problem, we note that by the variational principle
  \begin{multline}\label{eq:singweg}
    N_{l=0}\left(-\tau, \chi_0\left( \sqrt{p^2 + 1} - 1 - \kappa^C |\bx|^{-1} -v(|\bx|)\right) \chi_0 \right) \\
    \leq N_{l=0}\left(-\tau, \sqrt{p^2 + 1}- 1- \chi(\bx)(\kappa^C|\bx|^{-1}-v(|\bx|)) \right) 
  \end{multline}
  where $\chi$ denotes the charateristic function of the support of $\chi_o$.
  With the singularity gone, the result follows as in the subcritical
  case. Namely, similarly as in \eqref{eq:momentumcut} we cut in
  momentum space according to small and large momenta.  Again, by the
  variational principle, the right-hand side of \eqref{eq:singweg} is
  then bounded from above by
  $$
  N_{l=0}(-\const\tau, |\bp|- w(|\bx|) ) + N_{l=0}(-\const\tau, p^2 -
  w(|\bx|) ) ,
  $$
  where $w(r)= \const\chi(r)(\kappa^C r^{-1}+v(r))$.  The first
  term is estimated with the help of Daubechies' inequality
  \cite{Daubechies1983}
  $$
  N_{l=0}(-\tau, |\bp|-w(|\bx|)) \leq \tau^{-1/2}
  \tr_{l=0}(|\bp|-w(|\bx|))_-^{1/2} \leq \const \tau^{-1/2}
  \int_0^\infty w(r)^{3/2} \,dr
  $$
  with the latter integral being finite. For the second term we
  estimate $w(r) \leq \const r^{-1}$ and use that
  $$
  N_{l=0}(-\tau, p^2 - \const |\bx|^{-1}) \leq \const \tau^{-1/2}.
  $$
  This concludes the proof of Theorem \ref{thm:evhydc}.
\end{proof}

Our proof of Theorem~\ref{thm:evhydbr} in the critical Brown-Ravenhall
case is based on a reduction to the Chandrasekhar case.  The next lemma
compares the number of eigenvalues of the critical operators $
b_{1/2,l}(\kappa^B) $ with those of the two operators $
c_{l'}(\kappa^C_{l'} ) $ with $ l' = 0, 1$ and critical coupling
constants $ \kappa^C_0 = 2/\pi $ and $ \kappa^C_1 = \pi/2 $, cf.
\eqref{rem:critconst}.

\begin{lemma}\label{lem:reduction}
  There exists a constant such that for $l=0,1$ and all $\tau>0$ one
  has
  \begin{align*}
    N\left(-\tau,b_{1/2,l}(\kappa^B) \right)
	\leq
    \const \left[ N\left(-\tau, c_{0}(\kappa^C_0)
      \right) + N\left(-\tau, c_{1}(\kappa^C_1)
      \right) \right] .
\end{align*}
\end{lemma}

\begin{proof}
  We start with the observation that $ (\kappa_0^C)^{-1} +
  (\kappa_1^C)^{-1} = 2 (\kappa^B)^{-1} $. Using the explicit form of
  the reduced operators (cf.\ Appendix~\ref{wasserstoff}), this
  implies the identities
  \begin{equation}\label{eq:decompls}
\begin{split}
    b_{1/2,0}(\kappa^B) & = \kappa^B \left(
      (\kappa^C_0)^{-1} \phi_0 \tilde b_{0,0} \phi_0 +
      (\kappa^C_1)^{-1} \phi_1 \tilde
      b_{1,1} \phi_1 \right), \\
    b_{1/2,1}(\kappa^B) & = \kappa^B \left(
      (\kappa^C_0)^{-1} \phi_1 \tilde b_{0,1} \phi_1 +
      (\kappa^C_1)^{-1} \phi_0 \tilde b_{1,0} \phi_0 \right),
\end{split}
 \end{equation} 
 where the operators $ \tilde b_{l,\nu} $ are defined in $
 L^2(\mathbb{R}_+) $ through quadratic forms
 \begin{equation*}
   \langle f, \tilde b_{l,\nu} f \rangle := \int_0^\infty \frac{E(p)
     -1}{2 \phi_{\nu}(p)^2} |f(p)|^2 \, \rd p - \kappa^C_l
   \int_0^\infty \int_0^\infty \overline{f(p)}
   k^C_{l}(p,q) f(q) \,\rd
   p \,\rd q \, .
 \end{equation*}
 In case $\nu = 1 $ it hence follows from $2 \phi_{1}(p)^2 \leq 1 $
 that $\langle f, \tilde b_{l,1} f \rangle \geq \langle f,
 c_{l}(\kappa^C_l) f \rangle $.  In case $ \nu=0 $ we use the
 inequality 
 \begin{equation}
   \label{Rupert-Ungleichung}
   (E(p) - 1 )\phi_0(p)^{-2}\geq\sqrt{p^2+4}-2 = 2 (E(p/2)-1)
 \end{equation}
 which is most easily seen by writing both sides in terms of $E(p)$.
 It implies
 \begin{equation*}
   \langle f, \tilde b_{l,0} f \rangle \geq 2  \langle u f, c_{l}(\kappa_l^C)\,  u f \rangle  
 \end{equation*}
 where the unitary scaling transformation $ u $ is defined through $
 (u f)(p) := \sqrt{2} f(2p) $.  The proof is completed by the
 variational principle.
\end{proof}

We are now ready to give a 

\begin{proof}[Proof of Theorem~\ref{thm:evhydbr}. Critical case]
 The previous lemma implies that it suffices to show that for $ l = 0, 1 $
$$
N\left(-\tau, c_{l}(\kappa^C_l)
      \right)\leq \const \tau^{-1/2} .
$$
In case $ l = 0 $ this was established in \eqref{eq:l=0zaehlen}, and
the case $ l = 1 $ follows similarly with the analogue of
\eqref{eq:LiebYaunumber} given in Corollary \ref{lem:LiebYau}.
\end{proof}

%%%%%%%%%%%%%%%%%%%%%%%%%%%%%%%%%%%%%%%%%%%%%%%%%%%%%%%%%%%%%%%%%%%%%%%%%%%

\subsection{Sobolev inequality for the critical Brown-Ravenhall
  operator}\label{sec:sobolev}

Having studied the eigenvalues of $ B_\kappa $ in the previous
subsection, we now turn to integrability properties of its
eigenfunctions.  The $ L^q $-norm of two-spinors $ \psi $ is given by
$$ \| \psi \|_q := \left(\int_{\rz^3} |\psi(\bx)|^q \rd \bx\right)^{1/q}, $$
where the modulus, $ |\cdot | $, refers to the Euclidean norm in $
\mathbb{C}^2 $. For $q=2$ we drop the subscript.
We aim at proving the following
\begin{theorem}[\textbf{$L^q$-properties of
    eigenfunctions}]\label{thm:Lpeigen}
  Let $ 2 \leq q < 3 $. There exists a constant $ C_q < \infty $ such
  that for any $ \kappa \in (0,\kappa^B] $ and all $ \psi \in
  \mathfrak{Q}(B_\kappa) $ with $ \langle \psi, B_\kappa \psi \rangle
  \leq 0 $ one has $ \psi \in L^q $ with
\begin{equation}\label{eq:Lpeigen}
	\| \psi \|_q \leq C_q \, \| \psi \|
\end{equation}
\end{theorem}
Note that \eqref{eq:Lpeigen} applies, in particular, to eigenfunctions
of $ B_\kappa $ corresponding to negative eigenvalues.  The proof of
Theorem~\ref{thm:Lpeigen}, which is spelled out below, relies on a
Sobolev inequality for the massless atomic Brown-Raven\-hall operator
in $\gH $ given by
\begin{equation*}
  \bok  :=|\bp| - \frac{\kappa}{2}
  \left(|\bx|^{-1}+\omega_\bp\cdot\boldsymbol{\sigma}
  \ |\bx|^{-1} \ \omega_\bp\cdot\boldsymbol{\sigma} \right).
 \end{equation*}
 This operator is bounded below (in fact, non-negative) if and only if
 $ \kappa \leq \kappa^B $.

\begin{theorem}[\textbf{Sobolev inequality}]\label{thm:mainsob}
  For any $2\leq q<3$ there exists a constant $C_q>0$ such that for
  all $\psi\in \mathfrak{Q}(\bokb)$,
\begin{equation}\label{eq:mainsob}
  \| \psi \|_q^2 \leq C_q \, \left\langle\psi,\bokb
    \psi\right\rangle^{\theta}\, \|\psi\|^{2(1-\theta)}, \qquad \theta = 6(\tfrac 12 -\tfrac 1q).
\end{equation}
\end{theorem}

It is illustrative to compare \eqref{eq:mainsob} with the `standard'
Sobolev-Ga\-gliar\-do-Nirenberg inequalities,
\begin{equation}\label{eq:sob}
  \| \psi \|_q^2 \leq C_q' \, \left\langle\psi,|\bp| \psi\right\rangle^{\theta}\, \|\psi\|^{2(1-\theta)}, 
  \qquad \theta = 6(\tfrac 12 -\tfrac 1q), \qquad 2\leq q \leq 3,
\end{equation}
see, e.g., \cite[Thm. 8.4]{LiebLoss2001}. Hence Theorem
\ref{thm:mainsob} says that, if the endpoint $q=3$ is avoided, an
inequality of the same form remains true after subtracting the maximal
possible multiple of $\cU(|\bx|^{-1})$ from $|\bp|$. Moreover, one can
show that \eqref{eq:mainsob} does not hold with $q=3$, not even if the
$L^3$-norm is replaced by the weak $L^3$-norm.

Note that if $\kappa<\kappa^B$ then \eqref{eq:mainsob} with $\bok$
instead of $\bokb$ follows from \eqref{eq:sob} -- but with a constant
that deteriorates as $\kappa\to\kappa^B$. The main point is to derive
an inequality which holds uniformly in $\kappa$ up to and including
the critical constant. Our proof is based on the somewhat surprising
fact that the Brown-Ravenhall operator with coupling constant
$\kappa^B$ can be bounded from below by the Chandrasekhar operator
with \emph{smaller} coupling constant~$\kappa^C$.

Before we start the proof of \eqref{eq:mainsob}, we provide the
\begin{proof}[Proof of Theorem~\ref{thm:Lpeigen}] 
  The Sobolev inequality \eqref{eq:mainsob} implies
 \begin{align*} 
   \|\psi\|_{q}^2&\leq C_q \left\langle\psi,\bok\psi\right\rangle^\theta\|\psi\|^{2(1-\theta)}\leq C_q \big\langle\psi,\big[\bok-B_\kappa\big]\psi\big\rangle^\theta\|\psi\|^{2(1-\theta)}  \\
   &\leq C_q\|\bok-B_\kappa\|^\theta\|\psi\|^{2} .
  \end{align*}
  Tix showed \cite[Thm. 1]{Tix1997S} (see also Balinsky and Evans
  \cite{BalinskyEvans1998}) that the difference $ \bok-
  B_\kappa $ extends to a bounded operator with norm uniformly bounded for any $ \kappa \in (0,\kappa^B] $.
\end{proof}

\subsubsection{Comparison of critical operators}

The first step in the proof of the Sobolev inequality
\eqref{eq:mainsob} is a comparison of $\bok$ with the massless atomic
Chandrasekhar operator in $\gH$, which is given by
\begin{equation*}
    \cok  := |\bp|-\kappa|\bx|^{-1} .
\end{equation*}
It is bounded below if and only if $ \kappa \leq \kappa^C $. As
discussed in Appendix \ref{wasserstoff} those parts of $\bok$ and
$\cok$ in the subspace $\gH_{j,l,m}$ are unitarily equivalent to
operators $b^{(0)}_j(\kappa)$ and $c^{(0)}_l(\kappa)$ in $L^2(\rz_+)$,
which depend only on $j$ in the Brown-Ravenhall case and only on $l $
in the Chandrasekhar case.  For the comparison argument it is
important to note that the reduced operators $b^{(0)}_j(\kappa)$ and
$c^{(0)}_l(\kappa)$ are lower bounded for $ \kappa $ up to and
including the critical coupling constants $\kappa^B_j$ and
$\kappa^C_l$ respectively. They are defined in \eqref{rem:critconstbr}
and, as is explained there, exceed $ \kappa^B $ and $ \kappa^C $, if $
j \geq 3/2 $ or $ l \geq 1 $.

We begin by observing that all the critical operators
$b^{(0)}_j(\kappa^B_j)$ and $c^{(0)}_l(\kappa^C_l)$ have the same
`generalized ground state', namely $p$. The corresponding ground state
representation formula (in momentum space) is given in

\begin{lemma}[\textbf{Ground state representation}]\label{gsr}
  If $f\in \mathfrak{Q}(b^{(0)}_{j}(\kappa^B_j))$ and $g(p) = p f(p)$,
  then
  \begin{equation}\label{eq:gsr}
    \langle f,b^{(0)}_j(\kappa^B_j) f \rangle = 
    \frac{\kappa^B_j}{2}  \int_0^\infty \int_0^\infty 
    |g(p)- g(q)|^2 k^B_j(\tfrac12(\tfrac pq+\tfrac qp))
    \,\frac{\rd p}p\frac{\rd q}q.
  \end{equation}
  Similarly, if $f\in \mathfrak{Q}(c^{(0)}_{l}(\kappa^C_l))$ and $g(p)
  = p f(p)$, then
  \begin{equation}\label{eq:gsr2}
    \langle f,c^{(0)}_l(\kappa^C_l) f \rangle = 
    \frac{\kappa^C_l}{2}  \int_0^\infty \int_0^\infty 
    |g(p)- g(q)|^2 k^C_l(\tfrac12(\tfrac pq+\tfrac qp))
    \,\frac{\rd p}p\frac{\rd q}q.
  \end{equation}
  where $ k^B_j $ and $ k^C_l$ are defined in \eqref{eq:critconstjl}. 
\end{lemma}

\begin{proof}
  We write $k$ for one of the functions $k^B_j$ or $k^C_l$ and
  $\kappa$ for the corresponding constant $\kappa^B_j$ or
  $\kappa^C_l$. Expanding the square, we find
  \begin{align*}
    & \frac 12 \int_0^\infty \int_0^\infty |g(p)- g(q)|^2
    k(\tfrac12(\tfrac pq+\tfrac qp))
    \,\frac{\rd p}p\frac{\rd q}q \\
    = &\int_0^\infty | g(p)|^2 \left(\int k(\tfrac12(\tfrac
      pq+\tfrac qp))\,\frac{\rd q}q\right) \frac{\rd p}p - \int_0^\infty
    \int_0^\infty
    \overline{g(p)}k(\tfrac12(\tfrac pq+\tfrac qp))g(q)\frac{\rd p}p\frac{\rd q}q \\
    =& \int_0^\infty p | f(p)|^2 \left(\int k(\tfrac12(\tfrac
      pq+\tfrac qp))\,\frac{\rd q}q\right) \,\rd p - \int_0^\infty
    \int_0^\infty \overline{f(p)} k(\tfrac12(\tfrac pq+\tfrac qp))
    f(q) \,\rd p\,\rd q.
  \end{align*}
  By definitions \eqref{rem:critconstbr} and \eqref{rem:critconst} of
  $\kappa$ we have
  $$
  \int_0^\infty k(\tfrac12(\tfrac pq+\tfrac qp))\,\frac{\rd q}q = \kappa^{-1},
  $$
  which implies the assertion.
\end{proof}

Now we bound $\bokb$ from below by $\cokc$.

\begin{lemma}[\textbf{Comparison of critical operators}]\label{comp}
  There is a positive constant such that for any $\psi \in
  \mathfrak{Q}(\bokb)\cap \gH_{1/2,1}^\bot $
  \begin{align}\label{eq:comp}   
    \left\langle \psi, \bokb \psi \right\rangle \geq
    \const \left\langle \psi, \cokc \psi
    \right\rangle
  \end{align}
\end{lemma}

An inequality of the form \eqref{eq:comp} cannot hold in the subspace
$\gH_{1/2,1}$, since the right hand side is bounded from below by a
constant times $\langle \psi, |\bp| \psi \rangle$ while the left hand
side is not.

\begin{proof}
  By orthogonality it suffices to prove the inequality on each
  subspace $\gH_{j,l}$. First let $(j,l)=(1/2,0)$. We may also fix
  $m=\pm 1/2$ and choose $\psi\in \gH_{1/2,0,m}$. Its Fourier
  transform is of the form $\hat\psi(\bp) = p^{-1} f(p)
  \Omega_{\eh,0,m}(\omega_\bp)$, see Appendix \ref{app:J}. By the
  massless analog of \eqref{2.7a} one has
  $$
  \langle \psi, \bokb \psi \rangle
  = \langle f, b^{(0)}_{1/2}(\kappa^B) f \rangle.
  $$
  Setting $f(p) =: p g(p)$ we obtain in view of Lemma \ref{gsr}
  \begin{align*}
  	\langle f, b^{(0)}_{1/2}(\kappa^B) f \rangle
    & = \frac{\kappa^B}{2} \int_0^\infty \int_0^\infty |g(p)- g(q)|^2
    k^B_{1/2}(\tfrac12(\tfrac pq+\tfrac qp))
    \,\frac{\rd p}p\frac{\rd q}q \\
    & \geq (1+(2/\pi)^2)^{-1} \frac{\kappa^C}{2} \int_0^\infty \!\!\!
    \int_0^\infty |g(p)- g(q)|^2 k^C_0(\tfrac12(\tfrac pq+\tfrac qp))
    \,\frac{\rd p}p\frac{\rd q}q \\
    & \geq (1+(2/\pi)^2)^{-1} \langle f, c^{(0)}_0(\kappa^C) f \rangle 
     = (1+(2/\pi)^2)^{-1} \langle \psi, \cokc \psi
    \rangle
  \end{align*}
  Here we used that $0\geq Q_1$ and the massless analog of
  \eqref{eq:chandradec}. This proves the assertion on the subspace
  $\gH_{1/2,0}$. Now assume that $\psi\in
  \left(\gH_{1/2,0}\oplus\gH_{1/2,1}\right)^\bot$ and note that
  $$
  |\bp| \geq \frac{\kappa^B_{3/2}}{2}
  \left(|\bx|^{-1}+\omega_\bp\cdot\boldsymbol{\sigma}\ |\bx|^{-1}\
    \omega_\bp\cdot\boldsymbol{\sigma}\right)
  $$ 
  on that space. Here we used that $\kappa_j^B$ is monotone increasing
  in $j$, see Appendix \ref{wasserstoff}. We conclude that
  $$
  \langle \psi, \bokb \psi \rangle \geq
  \frac{\kappa^B_{3/2}- \kappa^B_{1/2}}{\kappa^B_{3/2}} \langle \psi,
  |\bp| \psi \rangle \geq \frac{\kappa^B_{3/2}-
    \kappa^B_{1/2}}{\kappa^B_{3/2}} \langle \psi, \cokc
  \psi \rangle ,
  $$
  proving the assertion.
\end{proof}

%%%%%%%%%%%%%%%%%%%%%%%%%%%%%%%%%%%%%%%%%%%%%%%%%%%%%%%%%%%%%%%%%%%%%%%%

\subsubsection{Proof of the Sobolev inequality}\label{subsec:Sobolev}
We are now ready to give a
\begin{proof}[Proof of Theorem~\ref{thm:mainsob}]
  By scaling, \eqref{eq:mainsob} is equivalent to the inequality
  $$
  \| \psi \|_q^2 \leq C_q' \, \left(
    \left\langle\psi,\bokb \psi\right\rangle +
    \|\psi\|^{2} \right).
  $$
  This, together with the triangle inequality, shows that it is enough
  to prove the inequality separately on the subspaces $\gH_{1/2,1}$
  and $\gH_{1/2,1}^\bot$. On the latter subspace, the claim follows
  immediately from Lemma \ref{comp} above and the Sobolev inequality
  for the critical Chandrasekhar operator \cite[Corollary
  2.5]{Franketal2008H}. We now reduce the claim for the subspace
  $\gH_{1/2,1}$ to that for $\gH_{1/2,0}$. For this purpose, we note
  that the helicity operator $\H = \omega_\bp \cdot {\boldsymbol
    \sigma}$, cf. \eqref{eq:defm}, commutes with $\bok$
  and, by \eqref{eq:trafospin}, maps $\gH_{j,l}$ into $\gH_{j,2j-l}$.
  Hence if $\psi\in\gH_{1/2,1}$ then by the Sobolev inequality on
  $\gH_{1/2,0}$
  $$
  \left\langle\psi,\bokb \psi\right\rangle + \|\psi\|^{2}
  = \left\langle \H\psi, \bokb \H\psi\right\rangle + \|\H
  \psi\|^{2} \geq \const \| \H \psi \|_q^2.
  $$
  By Lemma~\ref{lem:helicity} the helicity is bounded on $L^q(\rz^3,\cz^2)$. 
\end{proof}

%%%%%%%%%%%%%%%%%%%%%%%%%%%%%%%%%%%%%%%%%%%%%%%%%%%%%%%%%%%%%%%%%%%%%%%%%%%%%

\subsection{Estimates on the electric potential}\label{sec:twisted}

The goal of this subsection is to compare twisted and untwisted
electric potentials. We begin with an estimates for point charges and
then turn to smeared out charges.
\begin{lemma}\label{lemma:bu}
  Let $ l \geq 1 $ and $ \psi \in \gH_{j,l} $. Then
  \begin{equation}\label{eq:bu}
    \left| \left\langle \psi,\left( |\bx|^{-1} -
    \cU(|\bx|^{-1})\right)\psi\right\rangle \right| \leq
    \frac{\const}{l^2 } \langle\psi,\bp^2\psi\rangle.
  \end{equation}
\end{lemma}

\begin{proof}[Proof of Lemma \ref{lemma:bu}]
  By orthogonality it suffices to prove the assertion for $ \psi\in
  \gH_{j,l,m}$. Its Fourier transform is of the form $ \hat\psi(\bp) =
  f(p) p^{-1} \Omega_{j,l,m}(\omega_{\bp}) $, cf. Appendix
  \ref{app:J}, and we compute similarly as in \eqref{2.7a}
  \begin{align*}
    & \langle \psi , \left(|\bx|^{-1} - \cU(|\bx|^{-1})\right) \psi
    \rangle \notag \\ 
    & \quad = \frac{1}{\pi}
    \int_0^\infty \rd p \, \overline{f(p)} \int_0^\infty \rd q \, f(q)
    \left\{ \left[1 - \phi_0(p) \phi_0(q) \right]
    Q_l\left(\tfrac{1}{2} \left(\tfrac qp + \tfrac pq\right)\right)
    \right.\notag \\ &\quad \mkern230mu \left.  - \phi_1(p) \phi_1(q)
    Q_{2j-l} \left(\tfrac{1}{2} \left(\tfrac qp +\tfrac
    pq\right)\right)\right\} \\
& \quad = \frac1{2\pi} (A_1 + A_2) 
  \end{align*}
with
  \begin{align*} 
    A_1 := & \int_0^\infty \rd p \, \overline{f(p)} \int_0^\infty
    \rd q \, f(q) \sum_{\nu=0}^1 \left(\phi_\nu(p) - \phi_\nu(q) \right)^2 \,
    Q_l\left(\tfrac{1}{2} \left(\tfrac qp + \tfrac pq\right)\right) ,
    \\ A_2 := & 2 \int_0^\infty \rd p \, \overline{f(p)} \int_0^\infty
    \rd q \, f(q) \phi_1(p) \phi_1(q) \notag \\ & \mkern100mu
    \times\left[ Q_l\left(\tfrac{1}{2} \left(\tfrac qp + \tfrac
    pq\right)\right) - Q_{2j-l}\left(\tfrac{1}{2} \left(\tfrac qp +
    \tfrac pq\right)\right) \right].
  \end{align*}
  We estimate these terms separately.  For the first term we use
  \eqref{eq:phiabsch} and \eqref{psiabsch} together with Abel's
  argument to turn Hermitian integral operators into multiplication
  operators by means of the Schwarz inequality (see also \cite[Ineq.
  (6.9)]{LiebYau1988}). Since the $Q_l$ are positive, we obtain
  \begin{align*}
    A_1 \leq & \int_0^\infty \rd p \frac{|f(p)|^2}{E(p)^4}
    \int_0^\infty \rd q \left(\frac{p}{q}\right)^2 E(p)^2
    \sum_{\nu=0}^1 \left(\phi_\nu(p) - \phi_\nu(q) \right)^2 E(q)^2
    Q_l\left(\tfrac{1}{2} \left(\tfrac qp + \tfrac pq\right)\right)\\ 
	\leq & \frac{5}{8} \int_0^\infty \rd p
    \frac{|f(p)|^2}{E(p)^4} \int_0^\infty \rd q
    \left(\frac{p}{q}\right)^2 (p-q)^2 Q_l\left(\tfrac{1}{2}
    \left(\tfrac qp + \tfrac pq\right)\right) \notag \\ = &
    \frac{5}{8} \int_0^\infty \rd p \frac{|f(p)|^2}{E(p)^4} p^3
    \int_0^\infty \frac{\rd q}{q^2} (1-q)^2 Q_l\left(\tfrac{1}{2}
    \left(q + q^{-1}\right)\right).
  \end{align*}
  We now use the bounds $ p^3 / E(p)^4 \leq p^2 $ and, for $q\geq
  1$, $ (1-q)^2 \leq q^2 -1 $ which yield
  \begin{align*}
    \int_0^\infty \frac{\rd q}{q^2} (1-q)^2 Q_l\left(\tfrac{1}{2}
      \left(q + q^{-1}\right)\right) & = 2 \int_1^\infty \frac{\rd
      q}{q^2} (1-q)^2 Q_l\left(\tfrac{1}{2} \left(q +
        q^{-1}\right)\right) \\ 
	& \leq 4 \int_1^\infty \rd x \, Q_l(x)
    = \frac 4 {l(l+1)},
  \end{align*}
  where the last step involved \cite[324(18)]{Erdelyietal1954II}. Thus, 
  \begin{equation*}
    A_1 \leq \frac5{2 l(l+1)} \int_0^\infty \rd p \, p^2 |f(p)|^2 
    = \frac5{2 l(l+1)} \langle \psi,\bp^2 \psi \rangle.
  \end{equation*}
  We estimate the term $A_2$ similarly by the Schwarz inequality,
  \begin{align*}
    A_2 \leq & \, 2 \int_0^\infty \rd p |f(p)|^2 \, |\phi_1(p)|^2 \,
    \int_0^\infty \rd q \, \frac{p}{q} \left| Q_l\left(\tfrac{1}{2}
    \left(\tfrac qp + \tfrac pq\right)\right) -
    Q_{2j-l}\left(\tfrac{1}{2} \left(\tfrac qp + \tfrac
    pq\right)\right) \right| \notag \\ \leq & \, 4 \int_0^\infty \rd p
    |f(p)|^2 p^2 \int_1^\infty \frac{\rd q}{q} \left|
    Q_l\left(\tfrac{1}{2} \left(q + q^{-1}\right)\right) -
    Q_{2j-l}\left(\tfrac{1}{2} \left(q + q^{-1}\right)\right) \right|.
  \end{align*}
  Due to the pointwise monotonicity \eqref{eq:monotonicity}
  the difference inside the modulus is of definite
  sign.  Without loss of generality, we may therefore assume $ 2j=
  2l+1 $. Using the integral representation \eqref{eq:integral} we can
  bound
  \begin{align*}
    & \int_1^\infty \frac{\rd q}{q} \left[ Q_l\left(\tfrac{1}{2}
        \left(q + q^{-1}\right)\right) - Q_{l+1}\left(\tfrac{1}{2}
        \left(q
          + q^{-1}\right)\right) \right] \notag \\
     = &\int_1^\infty \rd z z^{-l-2} \left( z - 1 \right)
    \int_1^{\tfrac{1}{2}\left(z+z^{-1}\right)} \frac{\rd
      x}{\sqrt{x^2-1}} \frac{1}{\sqrt{1 - 2 xz + z^2}} \notag \\
    \leq& \frac{\pi}{\sqrt{2}} \int_1^\infty \rd z z^{-l-5/2}
    \left( z - 1 \right) = \frac{\pi}{\sqrt{2} (l +
      \tfrac12)(l+\tfrac32)} .
  \end{align*}
  Adding the estimates for $A_1$ and $A_2$ we arrive at \eqref{eq:bu}.
\end{proof}
Note that our proof shows that one can choose different powers of $|\bp|$
on the right hand side of \eqref{eq:bu}.

\begin{lemma}
  \label{lemma:Uv}
  There exists a constant such that for any electric potential $ v $
  of a spherically symmetric non-negative charge density
  \begin{equation*}
    \left| \left\langle \psi , \left( v - \cU( v) \right) \psi
    \right\rangle \right| \leq \const \, v(0) \, \langle \psi , \bp^2
    \psi \rangle .
    \end{equation*}
\end{lemma}
\begin{proof}
  We denote by $ \tau: \rz^3 \to [0,\infty) $ the spherically
    symmetric, non-negative charge density corresponding to $ v $,
    i.e., $ v(\bx) = \int \tau(\bx - \by) \, |\by|^{-1} \rd \by $.
    The Fourier transform of $ \tau $ obeys the estimates
    \begin{equation*}
      \left| \hat \tau(\bp) \right| = \sqrt{\frac{ 2}{ \pi \, \bp^2 }
      } \int_0^\infty r \, | \sin(|\bp| r) \, \tau(r) | \rd r \leq
      \frac{ v(0)}{(2\pi)^{3/2} \, |\bp| }
\end{equation*} 
By Fourier transform the scalar product on the left side of the
assertion becomes
$$
\left\langle \psi , \left( v - \cU( v) \right) \psi \right\rangle
 = \iint\!\hat\psi(\bp)^* \frac{\hat{\tau}(\bp - \bq)}{|\bp - \bq|^2} 
    \left(1-\Phi_0(\bp) \Phi_0(\bq)-\Phi_1(\bp)
    \Phi_1(\bq)\right)\hat\psi(\bq)  \rd \bp \rd\bq.
$$ Using Lemma~\ref{lem:abphi} below we estimate the absolute value of
the preceding expression from above by two terms, $B_1 $ and $ B_2
$. The first term can be further bounded as follows,
\begin{align*}
  \label{B1}
  B_1 & = \const \iint |\hat\tau(\bp - \bq)| |\hat\psi(\bp) |
  |\hat\psi(\bq) | \, \rd\bp \rd \bq \notag \\ & \leq \const v(0) \,
  \int \rd\bp \, |\hat\psi(\bp) |^2 \int
  \left(\frac{|\bp|}{|\bq|}\right)^{5/2} \frac{1}{|\bp - \bq|} \, \rd
  \bq \notag \\ & \leq \const v(0) \, \int |\hat\psi(\bp) |^2 \bp^2 \,
  \rd\bp ,
\end{align*}
where we use the Schwarz inequality in the second step.
The second term is estimated similarly
\begin{align*}
  B_2 & = \const \iint |\hat\tau(\bp - \bq)| \frac{\sqrt{|\bp | \,
  |\bq| }}{|\bp - \bq|} |\hat\psi(\bp) | |\hat\psi(\bq) | \, \rd\bp
  \rd \bq \notag \\ & \leq \const v(0) \, \int \rd\bp \,
  |\hat\psi(\bp) |^2 \int \left(\frac{|\bp|}{|\bq|}\right)^{2}
  \frac{\sqrt{|\bp | \, |\bq| }}{|\bp - \bq|^2} \, \rd \bq \notag \\ &
  \leq \const v(0) \, \int |\hat\psi(\bp) |^2 |\bp|^2 \, \rd\bp .
\end{align*}
\end{proof}

\section{Spectral shift from Schr\"odinger to Brown-Ravenhall
  operators\label{sec:BC}}

The main theme of this section is the (integrated) spectral shift,
i.e., the difference of sums of eigenvalues of the Brown-Ravenhall
or Chandrasekhar operator \begin{align*}
B[v] := \sqrt{p^2+1}-1- \cU(v), \qquad
   C[v] & := E(p)-1 - v, 
 \end{align*}
 (cf. \eqref{def:Ucal}) and the Schr\"odinger operator $ S[v] :=
 \tfrac12 \bp^2 - v $, all acting in the Hilbert space $\gH$ of
 two-spinors. We have set $ c = 1 $.

 Concerning the potential $v: \mathbb{R}^3 \to \mathbb{R} $ we will
 always assume that the above operators can be defined through the
 Friedrichs extension starting from $\gS(\mathbb{R}^3,\mathbb{C}^2) $.
 For example, the condition $ 0 \leq v(\bx) \leq \kappa^\# \,
 |\bx|^{-1} $ with $ \# = \mbox{\small \it B, C} $
 (cf.~\eqref{eq:couplingcrit}) ensures that the Brown-Ravenhall,
 respectively the Chandrasekhar operator are well-defined and bounded
 from below (see \cite{Evansetal1996} and \cite{Kato1966}).

 We assume throughout that the potential $v$ is radially symmetric
 which allows us to investigate the spectral shift on each subspace $
 \gH_{j,l} $ in the decomposition \eqref{eq:decomp} separately.  We
 write $\Lambda_{j,l}$ for the orthogonal projection onto $\gH_{j,l}$.
 For the reduced traces we use the notations
 \begin{equation*}
   \tr_{j,l}(A) := \tr(\Lambda_{j,l} A),
   \qquad \tr_j (A) :=\tr_{j,j+1/2}(A) + \tr_{j,j-1/2}(A).
 \end{equation*}

\subsection{Estimate on the spectral shift}\label{Sec:propshift}

One of the key observations in our proof of the Scott correction is
that the spectral shift between the one-particle Brown-Ravenhall and
the Schr\"odinger operator decreases sufficiently fast for high
angular momenta.

\begin{theorem}[\textbf{Spectral shift: Brown-Ravenhall case}\label{t:3}]
  There exists a constant $ C < \infty$ such that for any $ \kappa
  \leq \kappa^B $, any $ v: [0,\infty) \to [0,\infty) $ satisfying
  \begin{equation}\label{eq:vbound} v(r) \leq \kappa \, r^{-1},
  \end{equation}  any $\mu>0$ and any $j\in\nz_0 + 1/2$ one has
  \begin{equation}\label{eq:t3}
    \tr_{j} \left( \left[B[v]+\mu \right]_- - \left[ S[v] + \mu
      \right]_- \right) \leq C \,  \kappa^{4} \, j^{-2} .
  \end{equation}
\end{theorem}
We derive this result from a corresponding theorem for the
Chandrasekhar operator.  For a proof of the latter we need to
strengthen \cite[Thm. 2.1]{Franketal2008}.  In particular, we need to
consider $ C[v] $ for potentials $ v $ satisfying \eqref{eq:vbound}
also in case $\kappa^C < \kappa \leq \kappa^B $. Those operators are
\emph{not} densely defined in the Hilbert space~$ \gH $. However,
according to \eqref{rem:critconst} below, they \emph{are} densely
defined in the subspaces $ \gH_{j,l} $ with $j \geq 3/2$. Another new
aspect is that we trace the dependence on the coupling constant.
 
\begin{theorem}[\textbf{Spectral shift: Chandrasekhar
    case}]\label{prop:shiftbound}
  There exists a constant $C< \infty$ such that for all $l\in\nz_0$,
  $j=l\pm\tfrac12$, for all $\kappa$ satisfying
    \begin{equation*}
      \kappa\leq \left\{ \begin{array}{l@{$\quad$}l} \kappa^C &  \mbox{if $l=0$,} \\
          \kappa^B & \mbox{if $l\geq 1$}, 
		\end{array} \right.
	\end{equation*}
		 for all $\mu\geq0$ and for all $ v:
  [0,\infty) \to [0,\infty) $ satisfying   \eqref{eq:vbound},  one has
\begin{equation}\label{eq:shiftbound}
  0 \leq \tr_{j,l} \left( \left[C[v]+\mu \right]_- - \left[ S[v] + \mu \right]_-
    \right) \leq C  \frac{\kappa^{4}}{(l+\tfrac12)^2 } .
  \end{equation}
\end{theorem}  

One of the key points to be appreciated in the above theorems is an
effective cancellation in the differences in \eqref{eq:shiftbound} and
\eqref{eq:t3}. This can already be seen for Coulomb potentials
$v(r)=\kappa r^{-1}$, where
\begin{equation*}
  \tr_{j,l} \left[S_\kappa\right]_- = (2j+1) \frac{\kappa^2} 2
  \sum_{n=1}^\infty \frac1{(n+l)^2} ,
\end{equation*}
which does not decay at all as $j\to\infty$. Moreover, for fixed $j$
and $l$ the above trace vanishes only like $\kappa^2$ as $\kappa\to
0$.  It is rather remarkable that such cancellations occur uniformly
for all attractive potential $ v $ satisfying~\eqref{eq:vbound}.

The following proof of Theorem~\ref{prop:shiftbound} follows the ideas
of \cite[Thm.~2]{Franketal2008}.  It is not only included to render
the paper self-contained, but also to establish the above mentioned
improvement, which are important for the present paper.

\begin{proof}[Proof of Theorem~\ref{prop:shiftbound}]
  We note that both traces $\tr_{j,l} \left[C[v]+\mu \right]_-$ and
  $\tr_{j,l} \left[S[v]+\mu \right]_-$ are finite. This follows by the
  variational principle from the case $v(r)=\kappa r^{-1}$, cf.
  Theorem~\ref{thm:evhydc} in the Chandrasekhar case. Thus, for
  $l<3$ say, it is enough to show the claim for $\kappa$ in a
  neighborhood of $0$. More precisely, we can assume $\kappa\leq
  \tfrac1{\sqrt8} (l+\tfrac12)$ which covers all $\kappa\leq\kappa^B$
  for $l\geq3$.

  Moreover, by an approximation argument it is sufficient to consider
  $ \mu > 0 $ and bounded potentials $ v $, cf.~\cite{Franketal2008}.

  We denote by $ d_{j,l}$ the orthogonal projection onto the
  eigenspace of $C[v]$ corresponding to angular momenta $j, l$ and
  eigenvalues less or equal than $-\mu$.  The identity
  \begin{equation}
    \label{eq:banff}
    \tfrac12 p^2 = C_0 + \tfrac12 C_0^2 
  \end{equation}
  and the variational principle (cf.~\cite[Thm.~12.1]{LiebLoss2001}) imply
  \begin{equation}\label{eq:varprinc}
    0 \leq  2 \tr_{j,l} \left( \left[C[v] + \mu \right]_- - \left[S[v] + \mu \right]_- \right)
    \leq \tr \left[ C_0^2  d_{j,l} \right] .
  \end{equation}
  Using the eigenvalue equation and the bound \eqref{eq:vbound} on the
  potential we estimate this term further as follows.
\begin{equation*}
  \label{eq:4.1}
 \tr \left[ C_0^2 d_{j,l} \right] 
 \leq \tr_{j,l}  \left[ C[v] \right]_-^2 + \tr \left[ v^2  d_{j,l} \right]
 \leq \tr_{j,l} \left[ C_\kappa \right]_-^2 
 + \kappa^2 \tr\left[ |\bx|^{-2} d_{j,l} \right] .
\end{equation*}
Using Hardy's inequality and \eqref{eq:banff} 
\begin{equation*}
  \tr\left[ |\bx|^{-2} d_{j,l} \right]
   \leq 
   \left(l+\tfrac12\right)^{-2}  \tr \left[ \bp^2 d_{j,l} \right]
    =  \left(l+ \tfrac12\right)^{-2} \left(\tr \left[ C_0^2 d_{j,l} \right]
   + 2 \tr \left[ C_0 d_{j,l} \right]\right) .
\end{equation*}
Since $\kappa <  l + \tfrac12  $, the last two
estimates may be summarized as
\begin{equation}
  \label{eq:eiertanz}
  \tr \left[ C_0^2 d_{j,l} \right]  \leq \left( 1 - \frac{\kappa^2}{(l+\tfrac12)^2} \right)^{-1}
  \left(   \tr_{j,l} \left[ C_\kappa \right]_-^2 + \frac{2 \kappa^2}{ (l+\tfrac12)^2}  \tr\left[ C_0 d_{j,l} \right]\right) .
\end{equation}
We shall estimate the two terms on the right hand side separately.
From \cite[Lemma~3]{Franketal2008} we recall the following angular
momentum barrier inequality on $\gH_{j,l}$,
\begin{equation}\label{eq:ambi}
C_0 \geq 2\kappa r^{-1} \chi_{\{r\leq R_l(\kappa)\}},
\qquad
R_l(\kappa)= \tfrac1{8\kappa}(l+\tfrac12)^2 .
\end{equation}
(Here we use that $\kappa\leq \tfrac1{\sqrt8} (l+\tfrac12)$.) This implies
\begin{align*}
  \tr \left[ C_0 d_{j,l} \right] & \leq \kappa \tr \left[ |\bx|^{-1}
    d_{j,l} \right]
  \leq \frac12 \tr \left[ C_0 d_{j,l} \right] + \frac14 \tr\left[ w_l d_{j,l} \right] \notag \\
  & = \frac34 \tr\left[ C_0 d_{j,l} \right] - \frac14 \tr\left[ C[w_l]
    d_{j,l} \right]
\end{align*}
where $w_l(r) := 4 \kappa r^{-1} \chi_{\{r\geq R_l(\kappa)\}} $.
Hence, using the variational principle followed by Daubechies'
inequality \cite{Daubechies1983} (cf. also
\cite[Prop.~1]{Franketal2008})
\begin{align}
  \label{eq:daub1}
  \tr\left[ C_0 d_{j,l} \right] \leq  \tr_{j,l} \left[ C[w_l] \right]_-  & \leq \const (2l+1) \left( \int_0^\infty w_l(r)^{3/2} \,dr + \int_0^\infty w_l(r)^{2} \,dr \right) \notag \\
  & \leq \const \kappa^2.
\end{align}
In order to estimate the first term on the right hand side of
\eqref{eq:eiertanz} we use \eqref{eq:ambi} to obtain on $\gH_{j,l}$
$$
C_\kappa \geq \tfrac12 C_0 - \kappa r^{-1} \chi_{\{r\geq
  R_l(\kappa)\}} \geq \tfrac12 C[w_l] .
$$
with $w_l$ as above. Hence again by Daubechies' inequality
$$
\tr_{j,l} \left[ C_\kappa \right]_-^2 \leq \const (2l+1) \left(
  \int_0^\infty \!\! w_l(r)^{5/2} \,dr + \int_0^\infty \!\! w_l(r)^{3}
  \,dr \right) \leq \const \kappa^4 (l+\tfrac12)^{-2}.
$$
Combing this with \eqref{eq:daub1}, \eqref{eq:eiertanz}, and
\eqref{eq:varprinc} completes the proof.
\end{proof}

Having finished the proof of Theorem \ref{prop:shiftbound} it is easy
to give the
\begin{proof}[Proof of Theorem \ref{t:3}]
  Since the trace $\tr_{j} \left[B[v]+\mu \right]_-$ is finite
  according to Theorem~\ref{thm:evhydbr} we may assume that either
  $\kappa\leq \kappa^C$ and $j=1/2$, or else that $j\geq 3/2$. In this
  case, the claim essential boils down to
  Theorem~\ref{prop:shiftbound}. To see this, we note the identity
  \begin{equation}\label{eq:UHU}
    B[v]=\cU(C[v])=\tfrac12\left(U(\bp)^*C[v]U(\bp)+U(\bp)C[v]U(\bp)^* \right)
  \end{equation} 
  involving the unitary operator $U(\bp):=\Phi_0(\bp)+i\Phi_1(\bp)$
  (see also \eqref{eq:12a}).  Equality \eqref{eq:UHU} as well as the
  unitarity of $U(\bp)$ are easily derived from the fact that
  $\Phi_0^2(\bp)+\Phi_1^2(\bp)=1$.

  Even if $v$ satisfies \eqref{eq:vbound} only with a
  $\kappa^C<\kappa\leq\kappa^B$, identity \eqref{eq:UHU} remains valid
  on all subspaces $\gH_j$ with $j\geq 3/2$.  Hence by the concavity
  of the sum of negative eigenvalues \cite{Thirring1979} of $ B[v]+\mu
  $ one has for any $ \mu \geq 0 $
  \begin{align}
    \tr_j \left[B[v]+\mu \right]_- & \leq \frac{1}{2} \tr_j
    \left[U^*(\bp) C[v]U(\bp) +\mu \right]_- +
    \frac{1}{2} \tr_j  \left[U(\bp) C[v]U^*(\bp) +\mu \right]_- \notag \\
    & = \tr_j \left[C[v] +\mu \right]_- . \label{eq:sumfin}
  \end{align}
  By \eqref{eq:shiftbound} the trace in \eqref{eq:t3} is thus bounded
  from above by
  \begin{equation*}
    \tr_j \left( \left[C[v] +\mu \right]_- - \left[S[v] +\mu \right]_-
    \right) \leq \const \kappa^{4} j^{-2},
  \end{equation*}
  as claimed.
\end{proof}

%%%%%%%%%%%%%%%%%%%%%%%%%%%%%%%%%%%%%%%%%%%%%%%%%%%%%%%%%%%%%%%%%%%%%%%%%%

\subsection{Properties of the spectral shift}

In this subsection we discuss some properties of the spectral shift
$s(\kappa)$ defined in \eqref{eq:scott}.

\begin{lemma}[\textbf{Properties of the spectral shift}]
  The spectral shift $s$ is a continuous, non-negative function on $
  (0,\kappa^B]$ satisfying $ s(\kappa) = \cO(\kappa^2) $ as
  $\kappa\downarrow0$.
\end{lemma}

\begin{proof}
  According to \eqref{eq:t3a} and Theorem \ref{t:3} one has
  $$
  0 \leq s_{j}(\kappa) := \kappa^{-2} \tr_{j}\left(
    \left[B_\kappa\right]_- - \left[S_\kappa\right]_- \right) \leq
  \const \kappa^{2} j^{-2}.
  $$
  Therefore the sum $s(\kappa)=\sum_{j} s_{j}(\kappa)$ converges,
  is non-negative and satisfies the claimed asymptotic estimate as
  $\kappa\downarrow 0$. By the mini-max principle each eigenvalue
  depends continuously on $\kappa$.  Thus the continuity of their sum
  follows from the estimates in Theorem \ref{thm:evhydbr} and the
  Weierstra{\ss} criterion for uniform convergence.
\end{proof}

%%%%%%%%%%%%%%%%%%%%%%%%%%%%%%%%%%%%%%%%%%%%%%%%%%%%%%%%%%%%%%%%%%%%%%%%%
%%%%%%%%%%%%%%%%%%%%%%%%%%%%%%%%%%%%%%%%%%%%%%%%%%%%%%%%%%%%%%%%%%%%%%%%%

%%%%%%%%%%%%%%%%%%%%%%%%%%%%%%%%%%%%%%%%%%%%%%%%%%%%%%%%%%%%%%%%%%%%%%%%%%%%%%
%%%%%%%%%%%%%%%%%%%%%%%%%%%%%%%%%%%%%%%%%%%%%%%%%%%%%%%%%%%%%%%%%%%%%%%%%%%%%%

\section{Proof of the Scott correction} 
\label{sec:2}

The strategy of the proof of the main results is similar to the one
used for the Chandrasekhar operator \cite{Franketal2008}. We employ
the Schr\"o\-din\-ger operator as a regularization for the
relativistic problem, i.e., we will use it to eliminate the main
contribution to the energy (the Thomas-Fermi energy) and focus only on
the energy shift of the low lying states. For these the
electron-electron interaction plays no role and the unscreened problem
remains.  We define
$$
E^S(Z) := \inf\{ \cE^S_Z(\psi) \, | \, \psi \in \mathfrak{Q}^S_Z , \|
\psi \| = 1 \}
$$ 
to be the ground state energy in the
Schr\"odinger case,
\begin{align*}
  \cE^S_N(\psi) :=  \left\langle \psi, \left[\sum_{\nu=1}^N \left(\frac 12 \bp_\nu^2
        - Z |\bx_\nu|^{-1}\right) + \sum_{1\leq\mu<\nu\leq
        N}|\bx_\mu-\bx_\nu|^{-1} \right] \psi \right\rangle .
\end{align*}
It is defined on $ \mathfrak{Q}^S_N := \gH^S_N \cap
\mathfrak{S}(\mathbb{R}^{3N}, \mathbb{C}^{2^N} ) $,
where $ \gH^S_N := \bigwedge_{\nu=1}^N
\gH$ is the Hilbert space of anti-symmetric two-spinors.
We recall that we suppose neutrality, i.e., $N=Z$.

The asymptotics of the Schr\"odinger ground-state energy up to Scott correction reads
\cite{SiedentopWeikard1987O}
\begin{equation}\label{eq:SiedentopWeikard87}
  E^S(Z) = E_{\rm TF}(Z) + \tfrac{1}{2} \, Z^2 + \cO(Z^{47/24}).
\end{equation}
For our purpose this remainder estimate is sufficient. However, even
the coefficient of the $Z^{5/3}$-term in the asymptotic expansion is
known
\cite{FeffermanSeco1989,FeffermanSeco1990,FeffermanSeco1990O,FeffermanSeco1992,FeffermanSeco1993,FeffermanSeco1994,FeffermanSeco1994T,FeffermanSeco1994Th,FeffermanSeco1995}.

Our main result, Theorem~\ref{t2}, will follow from
\eqref{eq:SiedentopWeikard87} if we can show that in the limit $ Z \to
\infty $ the difference of the Schr\"odinger and Brown-Ravenhall
ground-state energy satisfies
\begin{equation}
  \label{reiheB}
  E^S(Z)-E^B_c(Z)=  s(Z/c) \, Z^2 + \cO(Z^{47/24})
\end{equation}
uniformly in $\kappa=Z/c\in (0,\kappa^B]$. We break the proof of this
assertion into an upper and lower bound.

%%%%%%%%%%%%%%%%%%%%%%%%%%%%%%%%%%%%%%%%%%%%%%%%%%%%%%%%%%%%%%%%%%%%%%%%

\subsection{Upper bound on the energy difference}

The Thomas-Fermi functional \eqref{eq:minimum} has a unique minimizer
$\varrho_Z$, the Thomas-Fermi density (Lieb and Simon
\cite{LiebSimon1977}). It scales as $\varrho_Z(\bx):=
Z^2\varrho_1(Z^{1/3}\bx)$. We set
\begin{equation}
  \label{eq:tf}
  \phi_\mathrm{TF}(\bx):= Z|\bx|^{-1} - \int_{\rz^3} {\varrho_Z(\by)\over|\bx-\by|}\,\rd\by,
\end{equation}
the Thomas-Fermi potential, and 
$$
L_\TF(\bx) := \int_{|\bx-\by|<R_Z(\bx)}{\varrho_Z(\by)\over|\bx-\by|}\,\rd\by,
$$
the exchange hole potential. Here $R_Z(\bx)$ is defined as the
(unique) minimal radius for which $ \int_{|\bx-\by|\leq R_Z(\bx)}
\varrho_Z(\by) \rd \by = \tfrac 12 $. The corresponding one-particle
operators -- self-adjointly realized in $\gH$ -- are
\begin{align*}
  S_\mathrm{TF} = S[\phi_\TF +L_\TF] ,
  \qquad
  B_\mathrm{TF} = B_c[\phi_\TF +L_\TF] .
\end{align*}
Here we use a notation analogous to that in \eqref{eq:BR}.

We shall express the many-particle ground-state energies $ E^S(Z) $
and $ E^B_c(Z) $ in terms of quantities involving the above
one-particle operators. In the Schr\"odinger case, this was
achieved in \cite{SiedentopWeikard1987O,SiedentopWeikard1989} in terms of
the Thomas-Fermi potential $\phi_\TF$. Our point in 
the proof of the following proposition is to replace $ \phi_\TF $ by the
exchange hole reduced potential $ \phi_\TF+L_\TF $.

\begin{proposition}\label{prop:schrupper}
Let $ J := \big[ Z^{1/9} \big] + \tfrac12 $. Then, as $ Z \to \infty $,
\begin{equation}\label{eq:schrupper}
  E^S(Z) = -\sum_{j=1/2}^{J-1} \tr_j \left[ S[Z|\bx|^{-1}] \right]_- 
  - \sum_{j=J}^{Z+1/2} \tr_j \left[ S_{\mathrm{TF}}\right]_- 
  - D(\varrho_Z,\varrho_Z) + O(Z^{47/24}).
\end{equation}
\end{proposition}

Since $\phi_\TF+L_\TF$ has a Coulomb tail, the trace
$\tr_j\left[S_\TF\right]_-$ is finite for each $j$, but not summable
with respect to $j$. It is therefore essential to restrict the second
sum to a finite number of angular momenta. However, the value of the
cut-off, $j\leq Z+ 1/2$, is not chosen optimally here, since for our
argument it is largely arbitrary.

\begin{proof}[Proof of Proposition \ref{prop:schrupper}]
  According to the correlation inequality \cite{Mancasetal2004}
  \begin{equation*}
    E^S(Z) \geq - \sum_{j=1/2}^{Z+1/2} \tr_j \left[ S_{\mathrm{TF}} \right]_- - D(\varrho_Z,\varrho_Z) .
  \end{equation*}
  Note that the $Z$ electrons can certainly be accommodated in the
  first $Z$ angular momentum channels (which is a very crude bound).
  Estimating $\phi_\TF+L_\TF$ from above by the Coulomb potential for
  small angular momenta, we obtain
  \begin{equation}\label{eq:mancas}
    E^S(Z) \geq -\sum_{j=1/2}^{J-1} \tr_j \left[ S[Z|\bx|^{-1}]\right]_- 
    - \sum_{j=J}^{Z+1/2} \tr_j \left[ S_{\mathrm{TF}}\right]_- - D(\varrho_Z,\varrho_Z) .
  \end{equation}
  Moreover, see \cite{SiedentopWeikard1987O,SiedentopWeikard1989},
  \begin{equation*}
    E^S(Z) \leq -\sum_{j=1/2}^{J-1} \tr_j \left[ S[Z|\bx|^{-1}] \right]_- 
    - \sum_{j=J}^{\infty} \tr_j \left[ S[\phi_{\mathrm{TF}}]\right]_- 
    - D(\varrho_Z,\varrho_Z) + \const Z^{47/24} .
  \end{equation*}
  Hence it suffices to prove that
  \begin{equation}\label{eq:goal}
    - \sum_{j=J}^{Z+1/2} \tr_j \left[ S_{\mathrm{TF}}\right]_-
    \geq-\sum_{j=J}^\infty\tr_j\left[S[\phi_{\mathrm{TF}}]\right]_--\const Z^{5/3}
  \end{equation}
  (Note that the lower bound in \cite{Franketal2008} contains an error
  by estimating \cite[Equation (43)]{Franketal2008} to generously.
  Really, only the first $Z$ lowest negative eigenvalues need to occur
  on the right hand side instead of all. In particular, there will be
  never more than $Z$ total angular momentum channels occupied. This
  fact is taken into account here yielding a suitable lower bound. The
  problem in \cite{Franketal2008} can be circumvented in exactly the
  same way.)  We decompose $L_\TF=L_<+L_>$ where
  \begin{equation*}
    L_< = \chi_{\{|\bx|<R\}}L_\TF,
    \qquad L_> = \chi_{\{|\bx|\geq R\}}L_\TF,
  \end{equation*}
  with a constant $R$ (independent of $Z$) to be chosen below. For
  $\varepsilon>0$ to be specified later we estimate using the
  variational principle for sums of eigenvalues
  \begin{multline}\label{eq:epsklauen}
    \tr_j \left[ S_{\mathrm{TF}}\right]_-\\
    \leq  \tr_j(\tfrac12(1-2 \varepsilon^2)p^2 - \phi_\TF)_- 
     + \varepsilon^2 \tr_j(\tfrac12 p^2 - \varepsilon^{-2} L_<)_- +
    \varepsilon^2 \tr_j(\tfrac12 p^2 - \varepsilon^{-2} L_>)_- .
  \end{multline}
  By the subsequent lemma the first and main term is bounded according
  to
  \begin{align*}
    & \sum_{j=J}^{Z+1/2} \tr_j(\tfrac12(1-2 \varepsilon^2)p^2 -
    \phi_\TF)_- - \sum_{j=J}^\infty \tr_j(\tfrac12 p^2 - \phi_\TF)_-
    \\ & \leq \tr(\tfrac12(1-2 \varepsilon^2)p^2 - \phi_\TF)_- -
    \tr(\tfrac12 p^2 - \phi_\TF)_- \leq \const \varepsilon^2 Z^{7/3}.
  \end{align*}
  For the second term on the right side of \eqref{eq:epsklauen} we use
  the Lieb-Thirring inequality~\cite{LiebThirring1976} to obtain
  \begin{align*}
    \varepsilon^2 \sum_{j=J}^{Z+1/2} \tr_j(\tfrac12 p^2 -
    \varepsilon^{-2} L_<)_-
    & \leq \varepsilon^2 \tr(\tfrac12 p^2 - \varepsilon^{-2} L_<)_- \\
    & \leq \const \varepsilon^{-3} \int L_<(\bx)^{5/2} \,\rd\bx \leq
    \const \varepsilon^{-3} Z^{2/3} .
  \end{align*}
  In the last inequality we used a bound of Siedentop and Weikard
  \cite[Proof of Lemma~2]{SiedentopWeikard1989}. It is at this point
  that $R$ is chosen. The penultimate inequality in \cite[Proof of
  Lemma~2]{SiedentopWeikard1989} asserts after scaling that $L_>(\bx)
  \leq \const |\bx|^{-1}$.  Hence by comparison with the exact
  hydrogen solution
  \begin{align*}
    \varepsilon^2 \sum_{j=J}^{Z+1/2} \tr_j(\tfrac12 p^2 -
    \varepsilon^{-2} L_>)_-
    & \leq \varepsilon^2 \sum_{j=1/2}^{Z+1/2} \tr_j(\tfrac12 p^2 - \varepsilon^{-2} \const |\bx|^{-1})_- \\
    & = \const \varepsilon^{-2} \sum_{j=1/2}^{Z+1/2} \sum_{n=1}^\infty
    \frac{2j+1}{(n+j-1/2)^2} \leq \const \varepsilon^{-2} Z .
  \end{align*}
  Choosing $\varepsilon=Z^{-1/3}$ all the error terms are
  $\cO(Z^{5/3})$, proving \eqref{eq:goal}.
\end{proof}

In the previous proof we used

\begin{lemma}
  For all $0<\varepsilon\leq 1/2$, as $Z\to\infty$,
  \begin{equation}\label{eq:phitf}
    \tr(\tfrac12(1-\varepsilon^2)p^2 - \phi_\TF)_-
    \leq \tr(\tfrac12 p^2 - \phi_\TF)_- + \const \varepsilon^2 Z^{7/3}.
  \end{equation}
\end{lemma}

Note that there are only a finite number of eigenvalues, since
$\phi_\TF$ decays like $|\bx|^{-4}$.

\begin{proof}
  Let $ d_\TF^\varepsilon$ be the projection onto the negative
  eigenvalues of $\tfrac12(1-\epsilon^2) p^2 -\phi_\TF$. Then, by the
  variational principle
  \begin{multline}
    \tr(\tfrac12(1-\varepsilon^2)p^2-\phi_\TF)_- -\tr(\tfrac12 p^2-\phi_\TF)_- \\
    \leq - \tr d_\TF^\varepsilon (\tfrac12(1-\varepsilon^2)p^2 -
    \phi_\TF) + \tr d_\TF^\varepsilon (\tfrac12 p^2 - \phi_\TF) =
    \tfrac{\varepsilon^2}{2} \tr d_\TF^\varepsilon p^2 .
  \end{multline}
  Hence the claim will follow, if we show that $\tr d_\TF^\varepsilon
  p^2 \leq \const Z^{7/3}$.  Note that $d^\varepsilon_\TF$ depends on
  both $ \varepsilon $ and $ Z $, and by rescaling one may get rid of
  the $ \varepsilon $ dependence at the expense of changing $ Z $. We
  may therefore assume that $\varepsilon=0$ and write $ d_\TF =
  d_\TF^0 $.

  Thus, it remains to prove
  \begin{equation}\label{eq:kinetic}
    \tr d_\TF p^2 \leq \const Z^{7/3}.
  \end{equation}
  Note that this says that the \emph{kinetic} energy is bounded by the
  order of the \emph{total} energy $\tr d_\TF (\tfrac12 p^2
  -\phi_\TF)$, which is well-known to be of order $Z^{7/3}$. Using
  that $\phi_\TF$ is bounded by a constant times
  $\min\{Z|\bx|^{-1},|\bx|^{-4}\}$ (see \cite{LiebSimon1977}) we get
  for any $R>0$
  \begin{align*}
    \tfrac12 \tr d_\TF p^2 & \leq \tr d_\TF\phi_\TF \\
    & \leq \const \left( \left( \int_{\{|\bx|<R\}} (Z|\bx|^{-1})^{5/2}\,\rd\bx \right)^{2/5} \left( \int  d_\TF(\bx,\bx)^{5/3} \,\rd\bx \right)^{3/5} \right. \\
    & \mkern350mu \left. + R^{-4} \int d_\TF(\bx,\bx) \,\rd\bx \right)
    .
  \end{align*}
  The Cwikel-Lieb-Rozenblum inequality (for a textbook presentation,
  see, e.g., \cite{Simon1979}) guarantees that
  \begin{align*}
    \int d_\TF(\bx,\bx) \,\rd\bx
    \leq\const\int\phi_\TF(\bx)^{3/2}\rd\bx = \const Z.
  \end{align*}
  Moreover, by the Lieb-Thirring inequality \cite{LiebThirring1976}
  \begin{align*}
    \int d_\TF(\bx,\bx)^{5/3} \,\rd\bx \leq \const \tr d_\TF p^2.
  \end{align*}
  We can estimate for any $\delta>0$
  \begin{align*}
    & \left( \int_{\{|\bx|<R\}} (Z|\bx|^{-1})^{5/2}\,\rd\bx \right)^{2/5} \left( \int  d_\TF(\bx,\bx)^{5/3} \,\rd\bx \right)^{3/5} \\
    & \qquad \leq \const Z R^{1/5} \left( \tr d_\TF p^2 \right)^{3/5} \\
    & \qquad \leq \delta \tr d_\TF p^2 + \const \delta^{-3/2} Z^{5/2}
    R^{1/2} .
  \end{align*}
  In summary, we have shown that
  \begin{equation*}
    \left(\tfrac 12 -\const\delta\right) \tr d_\TF p^2 \leq \const \left(\delta^{-3/2} Z^{5/2} R^{1/2} + R^{-4} Z \right).
  \end{equation*}
  Choosing $\delta$ small (of order one) and $R=Z^{-1/3}$ we obtain
  \eqref{eq:kinetic}.
\end{proof}

Next, we bound the many-particle ground state energy of the
Brown-Ravenhall operator from below by one-body quantities which match
the corresponding quantities in the Schr\"odinger case
\eqref{eq:schrupper}.
\begin{lemma}\label{lemma:blower}
  For all $ J \in \mathbb{N}_0 + 1/2 $ and $ Z \in \mathbb{N} $
  \begin{equation*}
    E^B_c(Z) \geq -\sum_{j=1/2}^{J-1} \tr_j \left[ B_c[Z|\bx|^{-1}]\right]_- 
    - \sum_{j=J}^{Z+1/2}\tr_j\left[B_{\mathrm{TF}}\right]_- - D(\varrho_Z,\varrho_Z) .
  \end{equation*}
\end{lemma}

\begin{proof}
  This follows by the same argument leading to \eqref{eq:mancas}.
\end{proof}

We are now ready to give a
\begin{proof}[Proof of Theorem~\ref{t2} -- first part]
  Choosing $ J= \big[ Z^{1/9} \big] + \tfrac12 $ and combining
  Proposition~\ref{prop:schrupper} and Lemma~\ref{lemma:blower} we
  obtain
  \begin{align}
    \label{eq:U1.1}
    E^S(Z)-E^B_c(Z) \leq & \sum_{j=1/2}^{J-1} \tr_j\left(\left[
        B_c[Z|\bx|^{-1}] \right]_-
      -\left[ S[Z|\bx|^{-1}]  \right]_- \right)  \\
    & + \sum_{j=J}^{Z+1/2} \tr_j \left( \left[ B_\mathrm{TF} \right]_-
      -\left[ S_\mathrm{TF} \right]_- \right)
    + \cO(Z^{47/24}). \notag
  \end{align}
  We note that by scaling $\bx\mapsto \bx/c$, the operators
  $S[Z|\bx|^{-1}]$ and $B_c[Z|\bx|^{-1}]$ are unitarily equivalent to
  the operators $ Z^2 \kappa^{-2}S_\kappa$ and
  $Z^2\kappa^{-2}B_\kappa$ where $\kappa =Z/c$. Similarly,
  $S_\mathrm{TF}$ and $B_\mathrm{TF}$ are unitarily equivalent to the
  operators $Z^2 \kappa^{-2}S[\kappa |\bx|^{-1}- \chi_c]$ and
  $Z^2\kappa^{-2}B[\kappa |\bx|^{-1}- \chi_c]$ acting in $\gH $, where
  $$
  \chi_c(\bx):= c^{-4} \int_{|\bx-\by|>c R_Z(c^{-1}\bx)}
  \frac{\varrho_Z(c^{-1}\by)}{|\bx-\by|} \,\rd\by .
  $$
  This implies that the first two terms on the right-hand side of
  \eqref{eq:U1.1}, which we denote by $\Sigma_1(Z,c) $ and
  $\Sigma_2(Z,c) $, can be rewritten as
  \begin{align*}
    \Sigma_1(Z,c) = & \, Z^2 \kappa^{-2} \, \sum_{j=1/2}^{J-1} \tr_j
    \left(\left[B_\kappa\right]_- -
      \left[S_\kappa\right]_-\right), \\
    \Sigma_2(Z,c) = & \, Z^2 \kappa^{-2} \, \sum_{j=J}^{Z+1/2} \tr_j
    \left(\left[B[\kappa|\bx|^{-1}- \chi_c]\right]_- - \left[S[\kappa
        |\bx|^{-1}- \chi_c]\right]_-\right) .
  \end{align*}
  Inequality \eqref{eq:t3a} and Theorem~\ref{t:3} guarantee that the terms
  in the first sum are non-negative and that the terms in both sums
  are bounded from above by a constant times $\kappa^{4} j^{-2}$ independently
  of $ Z $ and $ c $.  Therefore, the first sum can be bounded from
  above by an absolutely convergent series,
  \begin{equation*}
    \Sigma_1(Z,c) \leq Z^2 \kappa^{-2} \, \sum_{j=1/2}^{\infty} \tr_j 
    \left(\left[B_\kappa\right]_- -
      \left[S_\kappa\right]_-\right) 
    = Z^2 \, s(\kappa).
  \end{equation*}
  By the same token
  \begin{equation*}
    \Sigma_2(Z,c) \leq\const Z^2 \kappa^{2}\sum_{j=J}^\infty j^{-2}=\cO(Z^{17/9}),
  \end{equation*}
  uniformly in $c$. This concludes the proof of the upper bound on the
  energy difference.
\end{proof}

%%%%%%%%%%%%%%%%%%%%%%%%%%%%%%%%%%%%%%%%%%%%%%%%%%%%%%%%%%%%%%

\subsection{Lower bound on the energy difference} 

Similarly to \cite{SiedentopWeikard1987O} we define one-particle
density matrices $d^S$ and $d^B$ on $ \gH $ as sums
\begin{equation}
  \label{def:dB}
  d^\# = d^\#_< + d_>,
  \quad \# = S, B.
\end{equation}
The contribution of small total angular
momenta, $ d^\#_< = \sum_{l<L} d_{l}^\# $, is defined in Appendix
\ref{app:a1}. It comes from the eigenspinors of the atomic problems.
The contribution of large angular momentum,
$d_>= \sum_{l=L}^\infty d_{l}$, is defined in Appendix~\ref{app:a2}.
It corresponds to the Macke orbitals of \cite{SiedentopWeikard1987O}
and, in particular, coincides for the Schr\"odinger and
Brown-Ravenhall case.
The angular-momentum cut-off $ L $ will be chosen in a $ Z $-dependent
way, namely,
$$ L:=[Z^{1/12}] . $$

Important properties of the density matrices, whose construction is explained in more detail in Appendix~\ref{app:a}, are:
\begin{itemize}
\item The densities
  \begin{align*}
    \rho^\#(\bx) := \tr_{\mathbb{C}^2} \left(d^\# (\bx,\bx)\right),&
    \quad
    \rho^\#_{l}(\bx) := \tr_{\mathbb{C}^2} \left( d^\#_l (\bx,\bx)\right), \\
    \rho^\#_{<}(\bx) := \sum_{l<L} \rho^\#_{l}(\bx),& \quad
    \rho_{>}(\bx) := \sum_{l\geq L} \rho_{l}(\bx) .
  \end{align*}
  of $d^\#$, $ d^\#_l$, and $d>$ are all spherically symmetric.
\item The dimension of the ranges of the density matrices $ d^S $ and
  $d^B$ is at most $Z$, in particular $ \tr d^\# \leq Z $.  Moreover,
  \begin{equation}\label{eq:numberbr}
    \tr d^\#_l = \int \rho^\#_{l}(\bx) \, \rd \bx = 2(2l+1)(K-l),
    \qquad 0\leq l <L , 
  \end{equation}
   with $ K=[\const Z^{1/3}] $ and a suitable constant. 
\end{itemize}

For a lower bound on the ground state energy in the Schr\"odinger
case, we recall from \cite{SiedentopWeikard1987O} and
\cite[Proposition~4]{Franketal2008} the following
\begin{proposition}\label{prop:schrober}
  For large $ Z $, 
  \begin{equation*}
    E^S(Z) = \tr\left[ S[Z |\bx|^{-1}] \, d^S \right] +
    D(\rho^S,\rho^S) + \cO(Z^{47/24}) .
  \end{equation*}
\end{proposition}

To obtain an upper bound on the ground state energy in the
Brown-Ravenhall case, we use the reduced Hartree-Fock variational
principle.  It involves the density
\begin{equation*}
  \rho^B_{U}(\bx):= \tr_{\cz^2} \left( \cU_c(d^B) (\bx,\bx) \right)
\end{equation*}
of the twisted density matrix $\cU_c(d^B)$.

For further reference, we also set
$$
\rho^B_{U,l}(\bx) := \tr_{\mathbb{C}^2} \left( \cU_c(d^B_l)
  (\bx,\bx)\right) ,\quad \rho^B_{U,<}(\bx) := \sum_{l<L}
\rho^B_{U,l}(\bx) , \quad \rho_{U,>}(\bx) := \sum_{l\geq L}
\rho_{U,l}(\bx).
$$
Applying to \eqref{eq:1} the Hartree-Fock variational principle -- in
the strengthened version of Lieb \cite{Lieb1981V} (see also Bach
\cite{Bach1992}) -- and omitting the manifestly negative exchange
energy we arrive at
\begin{proposition}\label{prop:Hartree}
  For all $ Z $ and $ c $,
  \begin{equation*}
    E^B_c(Z) \leq \tr[B_c[Z |\bx|^{-1}]\,d^B] +
    D(\rho^B_U,\rho^B_U).
  \end{equation*}
\end{proposition}

Combining Propositions \ref{prop:schrober} and \ref{prop:Hartree} we
find
\begin{multline*}
  E^B_c(Z)- E^S(Z) \\
  \leq \tr[B_c[Z |\bx|^{-1}]d^B] - \tr[S[Z|\bx|^{-1}]d^S] +
  D(\rho_{U}^B-\rho^S,\rho_{U}^B+\rho^S) + \const Z^{47/24}.
\end{multline*}
Now we use the inequality $ \bp^2 \geq 2 c^2 ( E(p/c) -1 ) $ for the
kinetic energy corresponding to $d_>$. Morover we remark that
$D(\rho_{U,>}-\rho_>,\rho_{U,<}^B)\leq \cR_3$ and $D(f,g) \leq 0$, if
$f\leq0\leq g$. This yields
\begin{multline}\label{eq:brupper}
  E^B_c(Z)- E^S(Z) \leq \tr\left[B_c[\tfrac Z{|\bx|}]d^B_<\right] -
  \tr\left[S[\tfrac Z{|\bx|}]d_{<}^S\right] +
  \underbrace{\tr\left[\left(\tfrac Z{|\bx|} - \cU_c(\tfrac
        Z{|\bx|})\right) d_>\right]}_{=:\cR_1}\\
  +\underbrace{D(\rho_{U,>}-\rho_>,\rho_{U,>}+\rho_>)}_{=:\cR_2}
  +2\underbrace{D(\rho_{U,<}^B,\rho_{U}^B+\rho^S)}_{=:\cR_3} + \const
  Z^{47/24} .
\end{multline}
As we shall see, the first two terms will yield the Scott correction.
In the following subsections we prove that $\cR_1$, $\cR_2$, and
$\cR_3$, are relatively small remainder terms. Hence, we wish to
control the effects of the twisting operation $ \cU_c $, which stems
from the electronic projection, on the electrostatic Coulomb energy.

%%%%%%%%%%%%%%%%%%%%%%%%%%%%%%%%%%%%%%%%%%%%%%%%%%%%%%%%%%%%%%%%%%%%%%

\subsubsection{Controlling the electron projection for high angular momenta}

Our task in this subsection is to prove that for large angular
momenta, the twisted and untwisted electrostatic energy are
asymptotically equal.

We start by comparing the electric potential energy with or without
electron projection for large angular momentum. This will imply that
the term $\cR_1$ in \eqref{eq:brupper} is relatively small.
\begin{lemma}
  \label{electric-bubo1}
  In the limit $ Z \to \infty$ one has uniformly in $\kappa=Z/c\in
  (0,\kappa^B]$
  \begin{equation*}
    \int \left( \rho_{>}(\bx) - \rho_{U,>}(\bx) \right) \frac{ \rd\bx}{|\bx|} 
    = \tr \left[\left(|\bx|^{-1}-\cU_c(|\bx|^{-1})\right)d_> \right]   
    = \cO(Z^{11/12}).
  \end{equation*}
\end{lemma}
\begin{proof}
  Let $\{ \psi_\alpha \}$ stand for the Macke orbitals building up
  $d_>$ which we label by $ \alpha = (j,l,m,n) $; see
  \eqref{eq:d-groesser} and preceding equations in Appendix
  \ref{app:a2}.  By the scaling $\bx\mapsto\bx/c $ one has the
  relation
  \begin{equation*}
    \left\langle \psi_\alpha , \left[|\bx|^{-1}-\cU_c(|\bx|^{-1})\right]
      \psi_\alpha \right\rangle = c\, \left\langle \psi_\alpha^{(c)} ,
      \left[|\bx|^{-1}-\cU_1(|\bx|^{-1})\right] \psi_\alpha^{(c)}
    \right\rangle
  \end{equation*}
  where $ \psi_\alpha^{(c)}(\bx) := c^{-3/2} \psi_\alpha(\bx/c) $.
  Assuming that $ \alpha $ corresponds to a fixed (large) $ (j , l) $
  we may use Lemma~\ref{lemma:bu} to estimate the right-hand side by a
  constant times
  \begin{equation*}
    \frac{c }{l^2} \left\langle \psi_\alpha^{(c)} , \bp^2
      \psi_\alpha^{(c)} \right\rangle = \frac{1}{l^2\, c} \left\langle
      \psi_\alpha , \bp^2 \psi_\alpha \right\rangle .
  \end{equation*}  
  Using that $Z/c\leq\kappa^B$ we obtain the estimate
  \begin{equation*}
    \tr\left[\left(|\bx|^{-1}-\cU_c(|\bx|^{-1})\right)d_> \right]
    \leq\const\frac{\kappa^B}Z\sum_{l=L}^\infty\frac1{l^2}\sum_{j=l\pm1/2}\tr_{j,l}\left[\bp^2 d_>\right].
  \end{equation*}
  The proof is completed using Lemma~\ref{prop:kinetischeenergie} from
  Appendix \ref{app:b}.
\end{proof}

Next, we estimate the difference of Coulomb energies corresponding to
large total angular momenta. This shows that the term $\cR_2$
in \eqref{eq:brupper} may be neglected.

\begin{lemma}
  \label{electric-bubo2} 
  In the limit $ Z \to \infty$,one has uniformly in $\kappa=Z/c\in
  (0,\kappa^B]$
   \begin{equation*} 
     \cR_2=D(\rho_{U,>}-\rho_{>},\rho_>+\rho_{U,>}) = \cO(Z^{5/3}).
   \end{equation*}
\end{lemma}  
\begin{proof}
  We define $v:=(\rho_>+\rho_{U,>})*|\cdot|^{-1}$ to be the
  electric potential generated by $\rho_>+\rho_{U,>}$ which is
  obviously spherically symmetric and obeys
  \begin{align*}
    v(0) & = \tr \left[ d_> \left( |\bx|^{-1} + \cU_c(|\bx|^{-1})
    \right) \right] \notag \\ 
    & = 2 \tr \left[ d_> |\bx|^{-1} \right]
    - \tr \left[ d_> \left( |\bx|^{-1} - \cU_c(|\bx|^{-1}) \right)
    \right] .
  \end{align*} 
  According to \cite{SiedentopWeikard1987O} (see also
  \eqref{remindersw}) the first term on the right side is
  $\cO(Z^{4/3})$.  Moreover the second term is $\cO(Z^{11/12})$ by
  Lemma~\ref{electric-bubo1}, hence much smaller than the first term.
  Now,
  \begin{equation}
    \label{eq:sw1}
    D(\rho_>-\rho_{U,>},\rho_>+\rho_{U,>}) 
    = \tfrac12 \tr \left[d_>\left(v -\cU_c(v)\right)\right], 
  \end{equation}
  Decomposing the trace in \eqref{eq:sw1} into the orbitals
  contributing to $ d_> $ and scaling $ \bx \to \bx/c $ enables us to
  employ Lemma~\ref{lemma:Uv} to obtain the bound
  \begin{equation*}
    \tr \left[ d_> \left( v - \cU_c(v) \right) \right] \leq
    \frac{\const}{c^2} v(0) \, \tr\left[ d_> \bp^2 \right] .
  \end{equation*}
  This concludes the proof, since from \cite{SiedentopWeikard1987O}
  (see \eqref{remindersw}) we conclude that the trace on the
  right-hand side is $\cO(Z^{7/3})$.
\end{proof}

%%%%%%%%%%%%%%%%%%%%%%%%%%%%%%%%%%%%%%%%%%%%%%%%%%%%%%%%%%%%%%%%%%%%%%%%%%

\subsubsection{Contribution from low angular momenta to the Coulomb energy}

We now show that the term $\cR_3$ in \eqref{eq:brupper} is negligible.
\begin{lemma}
  \label{coulombsmallbr}
  In the limit $Z\to\infty$ one has uniformly in $ \kappa=Z/c\in
  (0,\kappa^B]$
  \begin{equation*}
    \cR_3=D(\rho_{U,<}^B,\rho_U^B+\rho^S)=\cO(Z^{11/6}\log Z).
  \end{equation*}
\end{lemma}

\begin{proof}	
  We first treat the term $ D(\rho_{U,<}^B,\rho_{U,>}+\rho_>) $.  By
  construction the densities $ \rho_{U,j}^B $ are spherically
  symmetric and satisfy according to \eqref{eq:numberbr}
  \begin{equation}\label{eq:numberbr1}
    \int \rho^B_{U,l}(\bx)\,\rd\bx  = \int\rho^B_{l}(\bx) \, \rd\bx 
    = 2(2l+1)(K-l), \quad 0\leq l <L.
  \end{equation} 
  Recalling the choice of $K$ and $ L$ we see that
  \begin{equation}\label{eq:numberbr2}
    \int \rho^B_{U,<}(\bx) \, \rd \bx
    = \cO(Z^{1/2}).
  \end{equation}
  It follows from \eqref{remindersw} and
  Lemma~\ref{electric-bubo1} that
  \begin{equation*}
    \int\frac{\rho_{U,>}(\bx) + \rho_>(\bx)}{|\bx|} \rd \bx 
    = \cO(Z^{4/3}).
  \end{equation*}
  Hence Newton's theorem \cite{Newton1972} yields
  \begin{equation*}
    D(\rho_{U,<}^B,\rho_{U,>}+\rho_>) \leq \frac{1}{2} \int\rho_{U,<}^B(\bx)
    \rd \bx \int \frac{\rho_{U,>}(\by)+\rho_>(\by)}{|\by|} \rd\by   
    =  \cO(Z^{11/6}).
  \end{equation*}
  In the remainder of the proof we are concerned with the term $
  D(\rho_{U,<}^B,\rho_{U,<}^B+\rho_<^S) $.  Noting that
  $$
  D(\rho_{U,<}^B,\rho_{U,<}^B+\rho_<^S) 
  \leq \frac 32 D(\rho_{U,<}^B,\rho_{U,<}^B)
  +\frac 12 D(\rho_{<}^S,\rho_<^S).
  $$
  and that according to \cite[Prop. 3.5]{SiedentopWeikard1987O}
  $D(\rho_{<}^S,\rho_<^S)=\cO(Z^{11/6})$, it suffices to consider
  $D(\rho_{U,<}^B,\rho_{U,<}^B)$.  We split the lowest angular
  momentum corresponding to $ l < 2 Z/c - 1/4=: l_0 $ off and define
  $$
  d_\vdash^B := \sum_{l \leq l_0} d_l^B, \qquad d_\dashv^B := \sum_{l
    > l_0}^{L-1} d_l^B,
  $$
  and
  $$
  \rho^B_{U,\,\vdash} := \tr_{\mathbb{C}^2} \left(
    \cU_c(d^B_\vdash)(\bx,\bx)\right), \qquad \rho^B_{U,\,\dashv} :=
  \tr_{\mathbb{C}^2} \left( \cU_c(d^B_\dashv)(\bx,\bx)\right).
  $$
  Note that in case $ l_0 < 0 $ there is no need for this procedure.
  Accordingly, we estimate
  \begin{equation*}
    D(\rho_{U,<}^B,\rho_{U,<}^B) \leq 2\, D(\rho^B_{U,\,\vdash}
    ,\rho^B_{U,\,\vdash} ) + 2\, D( \rho^B_{U,\dashv},
    \rho^B_{U,\dashv}) .
  \end{equation*}
  For an estimate of the second part corresponding to $l_0 < l<L$, we
  apply the following angular momentum barrier inequality
  \begin{equation}
    \label{eq:ang-skaliert}
    B_c(0) \geq \cU_c\left( \frac{2 Z}{|\bx|} \, \chi_{\{|\bx|
        \leq r_{l} \}} \right)
  \end{equation}
  on $ \gH_{j,l} $, where $ r_{l} = \left((l+1/2)^2 c^2 - 4
    Z^2\right)/ (4Z c^2) $ and $ l > 2 Z/c$.  This bound follows by
  applying $ \cU_1 $ to the inequality in
  \cite[Lemma~2.6]{Franketal2008} with $ R_l =
  [(l+1/2)^2-4\kappa^2]/(4\kappa)$ and scaling $ \bx \mapsto \bx/c $.

Inequality~\eqref{eq:ang-skaliert} implies
  \begin{align*}
    \tr\left[ \mathcal U_c\left(|\bx|^{-1} \right) d_l^B \right] &
    \leq \frac{1}{2Z} \tr\left[B_c(0)\, d_l^B \right] + \, \tr\left[
      \cU_c \left( |\bx|^{-1} \chi_{\{|\bx| > r_{l} \}} \right) d_l^B
    \right] \\
    & \leq \frac{1}{2} \tr\left[ \mathcal U_c\left(|\bx|^{-1} \right)
      d_l^B \right] + \frac{4 Z}{(l+1/2)^2 - 4 Z^2/c^2} \, \tr[d_l^B].
  \end{align*}
  Here the last inequality used the fact that eigenfunctions of $d_l^B
  $ are eigenfunctions of $B_c[Z|\bx|^{-1}] $ with negative
  eigenvalue. Now, note that 
  $$(l+1/2)^2 - 4 Z^2/ c^2 = (l+1/2+2Z/c)(l+1/2-2Z/c)\geq\const (l+1/2)^2$$ 
  since $l \geq l_0$.  Hence, using \eqref{eq:numberbr1} and summing
  over $l$ we obtain
  \begin{align*}
    \int\frac{\rho_{U,\dashv}^B(x)}{|\bx|}\rd \bx & = \sum_{l >
      l_0}^{L-1} \tr \left[ \mathcal U_c(|\bx|^{-1}) d_l^B \right]
    \leq \const Z \sum_{l = 0}^{L-1}(l+1/2)^{-2} \int\rho_{l}^B(\bx)\, \rd\bx \\
    & = \cO( Z^{4/3} \log Z ).
  \end{align*}
  Accordingly, Newton's theorem and \eqref{eq:numberbr2} yield
  \begin{align*}
    D( \rho^B_{U,\dashv}, \rho^B_{U,\dashv}) \leq \frac{1}{2}
    \int\rho_{U,\dashv}^B(\bx) \,\rd \bx \,
    \int\frac{\rho_{U,\dashv}^B(\bx)}{|\bx|} \,\rd \bx = \cO( Z^{11/6}
    \log Z ).
  \end{align*}

  Finally, we consider the contribution from $l \leq l_0 $. Note that
  then $ l \leq 2 \kappa^B - 1/4 < 2 $. We claim that the
  electrostatic energy corresponding to the electrons in this subspace
  is bounded by
  \begin{equation}
    \label{eq:smallsmall}
    D\big(\rho^B_{U,\,\vdash},\rho^B_{U,\,\vdash}\big) 
    \leq \const c K^2.
  \end{equation}
  Since by the choice of $l_0$ one has $ 2 Z/c \geq l + 1/4 \geq 1/4
  $, estimate \eqref{eq:smallsmall} will imply that $
  D\big(\rho^B_{U,\,\vdash},\rho^B_{U,\,\vdash}\big) \leq \const
  Z K^2 = \cO(Z^{5/3})$ and hence complete the proof of Lemma
  \ref{coulombsmallbr}. By scaling it suffices to prove
  \eqref{eq:smallsmall} for $ c = 1 $, which we will assume in the
  sequel. The Hardy-Littlewood-Sobolev inequality (see, e.g., \cite[Thm. 4.3]{LiebLoss2001})
  implies that
  \begin{equation}\label{eq:hls}
    D\big(\rho^B_{U,\,\vdash},\rho^B_{U,\,\vdash}\big) \leq \const
    \big\| \rho^B_{U,\,\vdash}\big\|^2_{6/5} .
  \end{equation}
  The triangle inequality together with the definition of $ \cU $ and
  \eqref{eq:trafospin} yields
  \begin{equation}\label{eq:triangle}
    \| \rho_{U, \,\vdash} \|_{6/5} \leq 
    \sum_{\alpha \in \mathcal{A}} \sum_{\nu=0,1} \| \Phi_\nu \psi_\alpha
    \|_{12/5}^2,
  \end{equation}
  where $ \{ \psi_\alpha | \alpha \in \mathcal{A} \} $ stands for the
  collection of normalized eigenfunctions building up $ d^B_{\vdash}
  $, i.e., the corresponding sum ranges over all indices $ (j,l,m,n)
  $.  We further estimate with the help of Lemma~\ref{lem:helicity} and Theorem~\ref{thm:Lpeigen},
    \begin{equation*} 
    \| \Phi_\nu \psi_\alpha \|_{12/5}^2 \leq \const \| \psi_\alpha
    \|_{12/5}^2 \leq \const .
  \end{equation*}
  This, together with
  \eqref{eq:hls}, \eqref{eq:triangle} and the fact that the number of
  indices in $\mathcal A$ is bounded by a constant times $K$ proves
  \eqref{eq:smallsmall}.
\end{proof}

\subsubsection{Finishing the proof} 
We repeat \eqref{eq:brupper}
$$
E^S(Z) - E^B_c(Z) \geq \tr[S[\tfrac Z {|\bx|}] d^S_<] - \tr[B_c[\tfrac Z{|\bx|}] d^B_<] - \cR_1 - \cR_2 - \cR_3 -\cR_4-\const Z^\frac{47}{24}
$$
By Lemmata \ref{electric-bubo1}, \ref{electric-bubo2}, and
\ref{coulombsmallbr} we have uniformly in $\kappa=Z/c\in (0,\kappa^B]$
$$
\cR_1 = \cO(Z^{23/12}), 
\quad \cR_2 = \cO(Z^{5/3}), \quad \cR_3 = \cO(Z^{11/6}\log Z),
$$
so these terms are of lower order than $Z^{47/24}$. Next we scale
$\bx\mapsto \bx/c$ and obtain
$$
\tr[S[Z |\bx|^{-1}] d^S_<] - \tr[B_c[Z |\bx|^{-1}] d^B_<] = Z^2
s(\kappa) - \cR_4
$$
where $s(\kappa)$ is introduced in \eqref{eq:scott} and
$$
\cR_4 := Z^2 \kappa^{-2} \sum_{l=0}^{L-1} (2l+1) \sum_{j=l\pm
  1/2} \sum_{n=K-l}^\infty \left(\lambda_n(s_{l}(\kappa)) -
  \lambda_n(b_{j,l}(\kappa) \right).
$$
By Theorem \ref{thm:evhydbr} there is a constant such that for all $0<\kappa\leq\kappa^B$
\begin{align*}
  0 \leq \cR_4 & \leq Z^2 \kappa^{-2} \sum_{l=0}^{L-1} (2l+1)
  \sum_{j=l\pm 1/2} \sum_{n=K-l}^\infty
  \left| \lambda_n(b_{j,l}(\kappa) \right| \\
  & \leq \const Z^2 \sum_{l=0}^{L-1} (2l+1) \sum_{n=K-l}^\infty (n+l)^{-2} \\
  & \leq \const Z^2 L^2 K^{-1} =
  \cO(Z^{11/6}).
\end{align*}
This concludes the proof of the lower bound and hence of our main
result. \qed \bigskip

%%%%%%%%%%%%%%%%%%%%%%%%%%%%%%%%%%%%%%%%%%%%%%%%%%%%%%%%%%%%%%%%%%%%%%%

\textbf{Acknowledgments:} We thank Elliott Lieb and Robert Seiringer
for supportive discussions. H.S. thanks the Departments of
Mathematics and Physics of Princeton University and R.F. and S.W. 
thank the Department of Mathematics of LMU Munich for hospitality
while parts of this work were done. We also thank Volker Bach for
stimulating questions resulting in several improvements.  The work has
been partially supported by the Deutscher Akademischer Austauschdienst, 
grant D/06/49117 (R.F.), the U.S.
National Science Foundation, grant PHY 01 39984 (H.S.), the Deutsche
Forschungsgemeinschaft, grant SI 348/13-1 (H.S.), and a Sloan
Fellowship (S.W.).

%%%%%%%%%%%%%%%%%%%%%%%%%%%%%%%%%%%%%%%%%%%%%%%%%%%%%%%%%%%%%%%%%%%%%%

\appendix

\section{Partial wave analysis \label{app:J}} 

For the convenience of the reader and for normalization of the
notation we gather some fact on the partial wave analysis of the
Brown-Ravenhall operator.

We denote by $Y_{l,m}$ the normalized spherical harmonics on the unit
sphere $\mathbb S^2$ (see, e.g., \cite{Messiah1969}, p. 421) with the
convention that $Y_{l,m}\equiv0$ if $|m|>l$, and we define for $ j\in
\mathbb{N}_0 + \tfrac{1}{2} $, $l \in \mathbb{N}_0 $, and $
m=-j,\ldots,j $ the spherical spinors
\begin{equation} 
  \Omega_{j,l,m}(\omega) :=        \begin{cases}
                \begin{pmatrix} \sqrt{\frac{j+m}{2j}} \,
                Y_{l,m-\eh}(\omega)\\ \sqrt{\frac{j-m}{2j}} \,
                Y_{l,m+\eh}(\omega)
                \end{pmatrix} & \text{if} \ j=l+ \eh,\\[4ex]
                \begin{pmatrix} -\sqrt{\frac{j-m+1}{2j+2}} \, 
                Y_{l,m-\eh}(\omega)\\ \sqrt{\frac{j+m+1}{2j+2}}\, 
                Y_{l,m+\eh}(\omega)
                \end{pmatrix} & \text{if} \ j=l-\eh.
        \end{cases}
\end{equation}
The set of admissible indices is $ \mathcal{I} := \{(j,l,m)\,
| \, j\in\nz-1/2,\ l=j\pm1/2,\; m=-j,...,j \} $. It is known that the
functions $\Omega_{j,l,m}$, $(j,l,m)\in\mathcal I$, form an
orthonormal basis of the Hilbert space $L^2(\mathbb
S^2;\mathbb{C}^2)$. They are joint eigenfunctions of $ \mathbf{J}^2 $,
$ J_3 $, and $\mathbf{L}^2 $ with eigenvalues given by $ j(j+1) $, $l(l+1)$,
and $m$.  The subspace $\gH_{j,l,m}$ corresponding to the joint
eigenspace of total angular momentum $ \mathbf{J}^2 $ with eigenvalue
$ j(j+1) $ and angular momentum $\mathbf{L}^2 $ with eigenvalue $
l(l+1) $ is then given by
$$\gH_{j,l,m} = {\rm span}\{ \bx \mapsto |\bx|^{-1} \, f(|\bx|) \,
\Omega_{j,l,m}(\omega_\bx) \ | \ f \in L^2(\mathbb{R}_+)\} 
$$ 
where $\omega_\bx:=\bx/|\bx|$. This leads to the orthogonal decomposition
\begin{equation}\label{eq:decompapp}
  \gH = \bigoplus_{j \in \nz_0 + \frac{1}{2} } \bigoplus_{l=j \pm 1/2} \gH_{j,l},
\qquad
\gH_{j,l} = \bigoplus_{m=-j}^j \gH_{j,l,m}, 
\end{equation}
of the Hilbert space of two spinors.

We note that the Fourier transform,
  \begin{equation}\label{eq:fourier}
    \hat\psi(\bp) := (2\pi)^{-3/2} \int_{\rz^3} e^{-i\bp\cdot\bx} \psi(\bx) \,\rd\bx ,
  \end{equation} 
leaves the spaces $\gH_{j,l}$ invariant. Namely, if we decompose $\psi$ according to \eqref{eq:decompapp},
$$
\psi(\bx)= \sum_{(j,l,m)\in \mathcal{I}} r^{-1} \psi_{j,m,l}(r)\,
\Omega_{j,l,m}(\omega_\bx),
$$
then
\begin{equation}\label{eq:decomp1}
  \hat\psi(\bp) = \sum_{(j,l,m)\in \mathcal{I}} p^{-1} \left(\mathcal F_l \psi_{j,m,l}\right)(p)\, \Omega_{j,l,m}(\omega_\bp)
\end{equation}
with the Fourier-Bessel transform
\begin{equation}\label{eq:FourierBessel}
(\mathcal F_l f) (p) = i^{-l} \sqrt{\frac2\pi} \int_0^\infty f(r) j_l(rp) rp \, \rd r .
\end{equation}
Here $j_l$ is a spherical Bessel function. Moreover,
\begin{equation*}
    \| \psi \|^2 = \!\sum_{(j,l,m)\in \mathcal{I}}\int_0^\infty
    |\psi_{j,m,l}(r)|^2\rd r = \| \hat \psi \|^2 = \!\sum_{(j,l,m)\in \mathcal{I}}\int_0^\infty
    |(\mathcal F_l\psi_{j,m,l})(p)|^2\rd p .
\end{equation*}

\section{Properties of the twisting operators\label{twisting}}
We define the helicity operator $\H ={\omega}_\bp
\cdot\boldsymbol{\sigma}$ on $ \gH $ by
\begin{equation}\label{eq:defm}
  \widehat{\H\psi}(\bp):=\boldsymbol{\sigma}\cdot\boldsymbol{\omega}_\bp\hat\psi(\bp).
\end{equation}
It follows from the pointwise identity
\begin{equation}\label{eq:trafospin}
  (\omega_\bp\cdot\boldsymbol{\sigma}) \Omega_{j,l,m}(\omega_{\bp}) =  -\Omega_{j,2j-l,m}(\omega_{\bp}) ,
\end{equation}
see, e.g., Greiner \cite[p. 171, (12)]{Greiner1990}, that $\H$ is an
isomorphism between $\gH_{j,l}$ and $\gH_{j,2j-l}$. Moreover, since $(
\boldsymbol{\sigma} \cdot \boldsymbol{a } ) \, (\boldsymbol{\sigma}
\cdot \boldsymbol{b}) = \boldsymbol{a } \cdot \boldsymbol{b} + i
\boldsymbol{\sigma} \cdot ( \boldsymbol{a } \times \boldsymbol{b} ) $
for any $ \boldsymbol{a }, \boldsymbol{b} \in \rz^3 $, we infer that $
H $ is an involution on $ \gH $, i.e., $ \H = \H^{-1} $.

We shall need to consider $\H$ on $L^p$ spaces with $p\neq 2$. The relevant properties are summarized in the next lemma, together with those of the operators
\begin{equation}\label{eq:propPhiHM}
  \widehat {\Phi_\nu \psi} (\bp) :=  \Phi_\nu(\bp ) \, \hat \psi(\bp) ,
\end{equation}
introduced in \eqref{fi1}. Note that while $ \Phi_0 $ acts trivially on the spin, $ \Phi_1 $ involves the helicity $\H$.

\begin{lemma}[\textbf{$\boldsymbol L^p $-properties of $\H$ and $ \Phi_\nu$}]
  \label{lem:helicity} 
The operators $\H$ and $ \Phi_\nu $, $ \nu = 0,1 $, extend to bounded operators from $ L^p(\rz^3, \cz^2) $ to $ L^p(\rz^3, \cz^2) $ for any $ p \in (1,\infty) $.
\end{lemma}

\begin{proof}
  The $L^p $-boundedness of $\H$ follows from that of the Riesz
  transformation, see \cite[Ch. II-III]{Stein1970}. Therefore, to
  prove the statement about the operators $\Phi_\nu$, it suffices to
  consider the operators $ \phi_\nu $ defined analogously as in
  \eqref{eq:propPhiHM} on $ L^2(\mathbb{R}^3) $. Since $ \bp \mapsto
  \phi_\nu(p) $ is smooth away from the origin and $p^k \partial^k
  \phi_\nu $ is bounded for $k=0,1,2$, the H\"ormander-Mihlin
  multiplier theorem \cite[Thm. IV.3]{Stein1970} implies that $
  \phi_\nu $ extend to bounded operators from $ L^p(\mathbb{R}^3) $ to
  $ L^p(\mathbb{R}^3) $ for any $ p \in (1,\infty) $.
\end{proof}

\begin{lemma}\label{lem:abphi}
For all $ \bp, \bq \in \rz^3 $
\begin{multline}
  1-\Phi_0(\bp) \Phi_0(\bq)-\Phi_1(\bp) \Phi_1(\bq) \\ = \frac{1}{2}
    \sum_{\nu=0}^1\left( \Phi_\nu(\bp) - \Phi_\nu(\bq) \right)^2 +
    \frac{1}{2}\left(\Phi_1(\bq) \Phi_1(\bp) - \Phi_1(\bp) \Phi_1(\bq)
    \right) .
\end{multline}
and furthermore
\begin{align*}
  & \left| \Phi_0(\bp) - \Phi_0(\bq) \right|^2 \leq \frac{|\bp -
    \bq|^2}{8 E(\bp)^2 E(\bq)^2} \\ & \left| \Phi_1(\bp) - \Phi_1(\bq)
    \right|^2 \leq \frac{|\bp - \bq|^2}{E(\bp) E(\bq)} \\ &
    \left|\Phi_1(\bq) \Phi_1(\bp) - \Phi_1(\bp) \Phi_1(\bq)\right|
    \leq \frac{\sqrt{|\bp| |\bq|} |\bp - \bq|}{E(\bp) E(\bq)}
\end{align*}
\end{lemma} 
\begin{proof}
The first equality is an immediate consequence of the definition of $
\Phi_0 $ and $ \Phi_1 $.  From this definition we also conclude by an
explicit calculation that
\begin{equation}\label{eq:phiabsch}
\left| \Phi_0(\bp) - \Phi_0(\bq) \right|^2 = ( \phi_0(\bp) -
\phi_0(\bq) )^2 \leq \frac{|\bp - \bq|^2}{8 E(\bp)^2 E(\bq)^2} .
\end{equation}
Moreover, for a proof of the next inequality we write
\begin{equation*}
  \left| \Phi_1(\bp) - \Phi_1(\bq) \right|^2 = ( \phi_1(\bp) -
  \phi_1(\bq) )^2 + \phi_1(\bp) \phi_1(\bq) | \omega_{\bp} -
  \omega_{\bq} |^2 ,
  \end{equation*}
and estimate the last two terms with the help of the inequalities
\begin{equation}
  \label{psiabsch}
  \left(\phi_1(\bp)-\phi_1(\bq)\right)^2 \leq
  \frac{(|\bp|-|\bq|)^2}{2E(\bp)^2 E(\bq)^2}\leq
  \frac{|\bp-\bq|^2}{2E(\bp)^2 E(\bq)^2} ,
\end{equation}
and
\begin{equation}\label{eq:help1}
  \phi_1(\bp) \leq \frac{1}{\sqrt{2}} \frac{|\bp|}{E(\bp)} \quad
  \mbox{and} \quad | \omega_{\bp} - \omega_{\bq}|^2 \leq \frac{|\bp -
  \bq|^2}{|\bp| |\bq|} .
\end{equation}
Finally, for a proof of the last inequality we use 
$$ \left|\Phi_1(\bq) \Phi_1(\bp) - \Phi_1(\bp) \Phi_1(\bq)\right| = 2
\phi_1(\bp) \phi_1(\bq) \left|\boldsymbol{\sigma} \cdot (\omega_{\bp}
\times \omega_{\bq}) \right| \leq 2 \phi_1(\bp) \phi_1(\bq)
|\omega_{\bp} - \omega_{\bq}| .
$$
Using again \eqref{eq:help1} concludes the proof of the third inequality.
\end{proof}

%%%%%%%%%%%%%%%%%%%%%%%%%%%%%%%%%%%%%%%%%%%%%%%%%%%%%%%%%%%%%%%%%%%%%%%%%%%

\section{Basics of relativistic hydrogenic operators\label{wasserstoff}}

In this section we collect -- following \cite{Evansetal1996} -- some
basic properties of the operators $B_\kappa$ and $C_\kappa$
which describe hydrogenic atoms in the Brown-Ravenhall respectively
Chan\-dra\-sek\-har model. 
For pedagogical reasons we first discuss their massless analogues,
\begin{equation}\label{eq:b0}
 \bok  :=|\bp| - \frac{\kappa}{2}
  \left(|\bx|^{-1}+\omega_\bp\cdot\boldsymbol{\sigma}
  \ |\bx|^{-1} \ \omega_\bp\cdot\boldsymbol{\sigma} \right), \quad 
 \cok  :=|\bp| - \kappa |\bx|^{-1} .
\end{equation}

 \subsection{Massless case} 
Expanding $ \hat\psi $ as in \eqref{eq:decomp1} and using  \eqref{eq:trafospin} 
yields \cite{Evansetal1996} the
following partial diagonalization of the massless operators,
\begin{align}
	 \langle \psi, B^{(0)}_\kappa  \psi\rangle  & = \sum_{(l,m,s)\in
  \mathcal{I}} \langle \mathcal{F}_l\psi_{j,m,l}, b^{(0)}_{j}(\kappa) \,
  \mathcal{F}_l\psi_{j,m,l} \rangle ,\notag \\
  	 \langle \psi, C^{(0)}_\kappa \psi\rangle  & = \sum_{(l,m,s)\in
  \mathcal{I}} \langle \mathcal{F}_l\psi_{j,m,l}, c^{(0)}_{l}(\kappa) \,
  \mathcal{F}_l\psi_{j,m,l} \rangle . \notag 
\end{align}
Here the operators $ b^{(0)}_{j}(\kappa) $ and $ c^{(0)}_{l}(\kappa) $
are densely defined in $ L^2(\mathbb{R}_+)$ through their quadratic
forms,
\begin{align*}  
  \langle f , b^{(0)}_{j}(\kappa) f \rangle & := \int_0^\infty p \,
  |f(p)|^2 \rd p - \kappa \int_0^\infty \int_0^\infty
  \overline{f(q)} \, k^B_{j}(q,p) \, f(p)\, \rd  q \, \rd p , \notag \\
  \langle f , c^{(0)}_{l}(\kappa) f \rangle & := \int_0^\infty p \,
  |f(p)|^2 \rd p - \kappa \int_0^\infty \int_0^\infty \overline{f(q)}
  \, k^C_{l}(q,p) \, f(p)\, \rd q \, \rd p , \notag
\end{align*} 
with maximal form domain denoted by $ \mathfrak{Q}(b^{(0)}_{j}(\kappa))
$ and  $ \mathfrak{Q}(c^{(0)}_{l}(\kappa)) $.
In the above expression, the integral
kernels $   k^B_{j} $ and $ k^C_{l} $ are given by
\begin{align}\label{eq:critconstjl}
     k^B_{j}(p,q) &:=\frac{1}{2\pi} \left[
        Q_{j-1/2}\left(\!\eh\left(\tfrac{p}{q}+\tfrac{q}p\right)\right)
        + Q_{j+1/2}\left(\!\eh\left(\tfrac p {q} +
        \tfrac{q}p\right)\right) \right] , \\
        k^C_{l}(p,q)& :=\frac{1}{\pi} \,
        Q_{l}\left(\!\eh\left(\tfrac{p}{q}+\tfrac{q}p\right)\right) ,
\end{align}
where $Q_l$ are Legendre functions of the second kind, i.e., 
\begin{equation}\label{eq:legendre}
Q_l(z) = \eh\int_{-1}^1 P_l(t) (z-t)^{-1}\rd t 
\end{equation}
with $P_l$ standing for Legendre polynomials; see Stegun
\cite{Stegun1965} for the notation and some properties of these
special functions.

It was proven in \cite{Evansetal1996} and \cite[Eq.~(5.33)]{Kato1966}
that \eqref{eq:b0} are self-adjoint and lower bounded if and only if $
\kappa \leq \kappa^\# $, $\# = B , C $, cf. \eqref{eq:couplingcrit}.
More can be said about the reduced operators $b_{j}^{(0)}(\kappa)$ and
$c_l^{(0)}(\kappa)$. They are lower bounded (in fact, non-negative) if
and only if
\begin{align}\label{rem:critconstbr}
  \frac{1}{\kappa} \geq \frac{1}{\kappa^B_j} 
  & :=\int_0^\infty k_j^B\left( \eh\left(\tfrac{1}{t}+t\right)\right)
  \,\frac{\rd t}{t}, \\
	\label{rem:critconst}
  \frac{1}{\kappa} \geq \frac{1}{\kappa^C_l}
  & :=\int_0^\infty k_l^C\left(\eh\left(\tfrac{1}{t}+t\right)\right)
  \,\frac{\rd t}{t}.
\end{align}
This follows by the same lines of reasoning as in \cite{Evansetal1996}. 

Since \cite[(8.4)]{Stegun1965}
$P_0(t)=1 $, $P_1(t)=t$,
we have
\begin{equation}\label{eq:q01}
Q_0(t)= \frac12 \log\frac{t+1}{t-1},
\quad
Q_1(t) = \frac t2 \log\frac{t+1}{t-1} - 1,
\end{equation}
such that $ \kappa^C_{0} = 2/\pi $, $ \kappa^C_{1} = \pi/2$ and thus
$\kappa^B_{1/2} = 2/(2/\pi+\pi/2) $.

The critical coupling constants $\kappa_j^B$ and $\kappa_l^C$ are
strictly increasing in $j$ and $l$ and, in particular, $ \kappa^B_{1/2}
= \kappa^B$ and $ \kappa^C_0 = \kappa^C$.  This follows from the
pointwise monotonicity
\begin{equation}\label{eq:monotonicity}
 Q_l(t) \geq Q_{l'}(t)
 \qquad\text{for} \ l' \geq l\ \text{and}\ t>1
\end{equation}
which, in turn, is evident from the integral representation
\begin{equation*}\label{eq:integral}
    Q_l(x)= \int_{x+\sqrt{x^2-1}}^\infty \frac{z^{-l-1}}{\sqrt{1-2xz+z^2}} \rd z  ,
    \qquad x>1;
\end{equation*}
see Whittaker and Watson \cite[p. 334, Chap. X, Sec.
3.2]{WhittakerWatson1927}.

\subsection{Massive case}\label{sec:decompbr}
Similarly as before one obtains the following partial diagonalization
of the massive hydrogenic Brown-Ravenhall and Chandrasekhar operators,
\begin{align} 
  \label{2.7a} 
  \langle \psi, B_\kappa \, \psi\rangle & = \sum_{(l,m,s)\in
  \mathcal{I}} \langle \mathcal{F}_l\psi_{j,m,l}, b_{j,l}(\kappa) \,
  \mathcal{F}_l\psi_{j,m,l} \rangle , \\
   \label{eq:chandradec}
  \langle \psi, C_\kappa\psi\rangle & = \sum_{(j,l,m)\in
  \mathcal{I}} \langle \mathcal{F}_l\psi_{j,m,l}, c_{l}(\kappa) \,
  \mathcal{F}_l\psi_{j,m,l} \rangle .
\end{align}
Here the operators $b_{j,l}(\kappa)$ and $c_{l}(\kappa)$ are
densely defined in $ L^2(\mathbb{R}_+)$ through their quadratic forms,
\begin{align}  
\label{eq:bls}  \langle f , b_{j,l}(\kappa) f \rangle  
   := \int_0^\infty (E(p) - 1 )  |f(p)|^2 \rd
  p - \kappa \int_0^\infty \int_0^\infty
  \overline{f(q)} \, k^B_{jl}(q,p) \, f(p)\, \rd  q \, \rd p , \\
\label{eq:cls}   \langle f , c_{l}(\kappa) f \rangle  
   := \int_0^\infty (E(p) - 1 )  |f(p)|^2 \rd
  p - \kappa \int_0^\infty \int_0^\infty
  \overline{f(q)} \, k^C_{l}(q,p) \, f(p)\, \rd  q \, \rd p 
\end{align} 
with maximal form domain denoted by $\mathfrak{Q}(b_{j,l}(\kappa))$
and $\mathfrak{Q}(c_{l}(\kappa))$, cf.~\cite{Evansetal1996}. In the
above expression, the integral kernel $ k^B_{j,l} $ depends, in
contrast to the massless case, on both $ j,l $, and is given by
\begin{equation*}
    k^B_{j,l}(p,q):=\frac{1}{\pi} \left[\phi_0(p)
        Q_l\left(\!\eh\left(\tfrac{p}{q}+\tfrac{q}p\right)\right)\phi_0(q)
        + \phi_1(p)Q_{2j-l}\left(\!\eh\left(\tfrac p {q} +
        \tfrac{q}p\right)\right)\phi_1(q)\right] .
\end{equation*}
The form \eqref{eq:bls} defines a self-adjoint semi-bounded operator
$b_{j,l}(\kappa)$, if and only if $\kappa\leq\kappa^B_j$ (Evans et al.
\cite{Evansetal1996}). In fact $b_{j,l}+c^2$ is positive (Tix
\cite{Tix1998}). A trivially modified argument shows that
\eqref{eq:cls} defines a self-adjoint semi-bounded operator
$c_l(\kappa)$, if and only if $\kappa\leq\kappa^C_l$.

In fact the semiboundedness of the massive cases and the massless
cases are equivalent, since the differences of the massive and
massless forms are bounded (Tix \cite[Thm.~1]{Tix1997S}).

%%%%%%%%%%%%%%%%%%%%%%%%%%%%%%%%%%%%%%%%%%%%%%%%%%%%%%%%%%%%%%%%%%%%%%%%%%%%%

\section{Critical Chandrasekhar operator on a finite
  domain\label{sec:liebyau}}

Lieb and Yau \cite{LiebYau1988} have shown that the critical
Chandrasekhar operator $ |\bp| - \kappa^C |\bx|^{-1} $ when restricted
to a ball has only discrete spectrum with eigenvalues accumulating at
infinity at the rate predicted by the semiclassical result for $ |\bp|
$ alone. This is remarkable since the semiclassical phase-space volume
corresponding to $ |\bp| - \kappa |\bx|^{-1} $ is infinite.

We aim to prove an analogous result for the Chandrasekhar operator
restricted to the fixed angular momentum subspace corresponding to $ l
= 1 $ and finite domain.  In the proof of Theorem~\ref{thm:evhydbr} it
is essential to handle coupling constants which are larger than $
\kappa^C $, all the way up to and including $ \kappa_1^C $.

In order to define the above operator we consider for $R>0$ and $ l
\in \nz $ the Hilbert space
$$
\mathfrak{F}_l(R) := \left\{ f \in L^2(0,\infty) \, \big| \,
  \left(\mathcal{F}_l^{-1} f\right)(r) = 0 \; \mbox{for all $ r \geq R
    $} \right\} ,
$$
where $ \mathcal{F}_l $ denotes the Fourier-Bessel transformation,
cf.~\eqref{eq:FourierBessel}.  The quadratic form given by $\langle f
, c_l^{(0)}(\kappa) f \rangle $ with domain $\mathfrak{F}_l(R) \cap
\mathfrak{Q}(c_l^{(0)}(\kappa))$ defines for all $ \kappa \leq
\kappa_l^C $ a self-adjoint, non-negative operator in $
\mathfrak{F}_l(R) $ which we will denote by $ c_l^{(0)}(\kappa,R) $.

\begin{lemma}\label{lem:LiebYau}
  Let $ l \in \nz $. There exists some constant such that for all $ R
  > 0 $, $ \mu > 0 $, and $ \kappa \leq \kappa_l^C $
\begin{equation}\label{eq:LiebYau}
 \tr \left( c_l^{(0)}(\kappa,R) - \mu \right)_- \leq \const \, \mu^2 \, R .
\end{equation}
\end{lemma}
We have not tried to track the $ l $-dependence of the constant, since
the cases $ l = 0, 1 $ will be enough for our purpose.

\begin{proof}
  For a proof of \eqref{eq:LiebYau} we basically follow the argument
  in \cite{LiebYau1988}. The starting point is the following reduction
  to a simpler variational problem involving only functions. Namely,
  for any non-negative function $h:\rz_+ \to \rz_+$, let
  $$
  t(p) := \frac{\kappa_l^C}{\pi h(p)} \int_0^\infty
  Q_l\big(\tfrac{1}{2}(\tfrac{p}{q}+ \tfrac{q}{p})\big) h(q) \,\rd q
  \,.
  $$
  Then
  \begin{equation}\label{eq:lyvarprinc}
    - \tr \left( c_l(\kappa_l^C,R) - \mu \right)_- 
    \geq \inf \left\{ \int_0^\infty \sigma(p) \left( p - \mu - t(p) \right)\rd p \, \big| \, 0 \leq \sigma \leq M_l \right\}
  \end{equation}
  where $M_l := R \, \sup_{r>0} (2/\pi) r^2 j_l^2(r) $.  The proof of
  \eqref{eq:lyvarprinc} is analogous to the one of
  \cite[Eq.~(7.8)]{LiebYau1988}. We merely replace the Fourier
  transformation in $ \rz^3 $ by the Fourier-Bessel transformation $
  \mathcal{F}_l $ in $ \rz_+ $.

  {F}rom now on we assume that $l\geq1$ and comment on the necessary
  changes in case $l=0$ at the end.  We choose $h$ of the form
  \begin{equation*}
    h(p) = \left\{ 
      \begin{array}{lr} p^{-1} - (A/2) p^{-2} & \mbox{if $ p > A$,}\\
        (2A)^{-1} & \mbox{if $ p \leq A $.}
      \end{array} \right.
  \end{equation*}
  Below we shall show that the constant $A$ can be picked in such a
  way that for some $\delta>0$
  \begin{equation}\label{eq:claimt}
    p-\mu - t(p) \geq \left\{\begin{array}{lr} 0 & \mbox{if $ p \geq
          \delta^{-1}A $,} \\ - \const A^{-1} \mu^2 & \mbox{if $ p <
          \delta^{-1}A $.}
      \end{array} \right.
  \end{equation}
  In view of \eqref{eq:lyvarprinc} this will prove the result, since
  then
  \begin{align*}
    \inf\left\{ \int_0^\infty \sigma(p) \left( p - \mu - t(p)
      \right)\rd p \, \big| \, 0 \leq \sigma \leq M_l \right\}
    & \geq - \const A^{-1} \mu^2 M_l \int_0^{\delta^{-1}A}\!\!\rd p \\
    & = - \const\delta^{-1} \mu^2 M_l .
  \end{align*}

  To prove \eqref{eq:claimt} we recall that $ \int_0^\infty
  Q_l\big(\tfrac{1}{2}(t+\tfrac1t)\big)\,\tfrac{dt}t = \pi
  (\kappa_l^C)^{-1} $, cf.  \eqref{rem:critconst}, and hence by a
  straightforward calculation
  \begin{align*}
    p-t(p) & = \frac{p}{2} \int_0^\infty
    Q_l\big(\tfrac{1}{2}(t+\tfrac1t)\big)
    \left(\tfrac1t - \tfrac{h(tp)}{h(p)}\right) \,\rd t  \\
    & = \frac p2 \begin{cases}
      \frac{A/2p}{1-A/2p} \left( F(1) - F(A/p) \right) & \text{if} \ p \geq A ,\\
      \left( - F(1) + F(p/A) \right) & \text{if} \ p < A .
    \end{cases}
  \end{align*}
  Here for $0<s\leq 1$ we have set
  \begin{equation*}
    F(s) := \int_0^{s} Q_l\big(\tfrac12(t+\tfrac1t)\big)\left(\tfrac1t-\tfrac1s \right)^2
    \,\rd t \,.
  \end{equation*}
  Since $Q_l(\tau) \leq Q_1(\tau)$, which vanishes like a constant
  times $\tau^{-2}$ as $\tau\to\infty$, one has $F(s)\to 0$ as $s\to
  0$.  Choosing $\delta\in(0,1)$ such that $F(s)\leq \tfrac 12 F(1)$
  for all $0<s\leq\delta$, we have shown that for all $p \geq
  \delta^{-1} A$ one has
  \begin{align*}
    p-t(p) \geq \frac{A\ F(1)}{8(1-A/2p)} \geq A \ \frac{F(1)}{8} \,.
  \end{align*}
  For $A\leq p < \delta^{-1} A$ we use the monotonicity, $dF/ds\geq
  0$, to bound
  $$    p-t(p) \geq 0.$$
  Finally, for $0\leq p <A$ we drop the term $F(p/A)\geq 0$ to obtain
  $$
  p-t(p) \geq - \frac p2 F(1) \geq - A \frac{F(1)}2.
  $$ 
  Choosing $A:=8\mu/F(1)$ yields the claimed inequality
  \eqref{eq:claimt}.

  In case $l=0$, the function $h$ can be chosen as before. However,
  the corresponding expressions $ F(1) - F(s) $ should be interpreted
  as a single integral, and estimated with slightly more care.
\end{proof}

\begin{corollary}\label{cor:LiebYau}
  Let $l\in \nz$. Then there exists some constant such that for all
  $0<\kappa\leq\kappa_l^C$, all $\mu>0$ and all functions $\chi$ on
  $\rz_+$ which satisfy $\chi>0$ on $[0,R)$ and $\chi\equiv 0$ on
  $[R,\infty)$ for some $R>0$, one has:
  $$
  N_l(0,\chi\left(|\bp| - \kappa |\bx|^{-1}-\mu\right)\chi) \leq
  \const \mu R .
$$
\end{corollary}

\begin{proof}
  The variational principle implies that
  $$
  N_l(0,\chi\left(|\bp| - \kappa |\bx|^{-1}-\mu\right)\chi) \leq
  N(\mu,c_l^{(0)}(\kappa,R)) .
  $$
  Indeed, if $ \mathcal{V}_l $ is the negative spectral subspace of
  $\chi\left(|\bp|-\kappa |\bx|^{-1}-\mu\right)\chi$ with fixed $ l $,
  then any $ f \in \mathcal{F}_l \chi \mathcal{V}_l \subset
  \mathfrak{F}_l(R) $ satisfies $ \langle f , ( c_l^{(0)}(\kappa,R) -
  \mu )\, f \rangle < 0 $.

  The assertion now follows from
  $$ N(\mu,c_l^{(0)}(\kappa,R)) \leq \const \mu R .
  $$ 
  For a proof, we note that the elementary inequality
  $\chi_{(-\infty,\mu)}(E)\leq\frac{(E-\lambda)_-}{\lambda-\mu}$,
  valid for any $ \mu < \lambda $, together with
  Lemma~\ref{lem:LiebYau} implies that
  $$ 
  N(\mu,c_l^{(0)}(\kappa,R)) \leq (\lambda - \mu)^{-1} \tr (
  c_l^{(0)}(\kappa,R) - \lambda)_- \leq \const (\lambda - \mu)^{-1}
  \lambda^2 R .
  $$
  The proof is completed by optimizing over $\lambda$.
\end{proof}

%%%%%%%%%%%%%%%%%%%%%%%%%%%%%%%%%%%%%%%%%%%%%%%%%%%%%%%%%%%%%%%%%%%%%%%%%%

\section{The trial density matrix\label{app:a}}

In this section we define the density matrices $d^S$ and $d^B$ that we
use to bound the Schr\"odinger energy, respectively the
Brown-Ravenhall energy, from above. Both density matrices are split
into two parts corresponding to low and high angular momenta
\begin{equation*}
  d^S  :=  d_<^S + d_> ,  \qquad 
  d^B  :=  d_>^B + d_> .
\end{equation*}
Low angular momenta correspond to orbits whose perinucleon is close to the
nucleus, while high angular momenta ensure that the orbits are
never close to the nucleus. We will cut between these two at
$L:= [Z^{1/12}]$.

\subsection{Low angular momenta\label{app:a1}}

In the vicinity of the nucleus the nuclear attraction dominates the
interaction with the other electrons. This motivates to choose the
orbitals as the ones of the Bohr atom, i.e., as the eigenfunctions of
the unscreened operator with nuclear charge $Z$. The corresponding
density matrices $d_<^\#$ are of the form
$$
d_<^\# = \sum_{l=0}^{L-1} d_j^\#,
\qquad
d_l^\# = \sum_{j=l\pm1/2,\ j\geq 1/2} d_{j,l}^\#
$$
and
$$
d_{j,l}^\#= \sum_{m=-j}^j \sum_{n=1}^{K-l} |\psi_{j,l,m,n}^\#\rangle\langle\psi_{j,l,m,n}^\#| .
$$
Here $K= [\const Z^{1/3}]$ with some positive constant, i.e., on the
order of the last occupied shell of the Bohr atom. We now turn to the
definition of the orbitals $\psi_{j,l,m,n}^\#$ for which we consider
the cases $\#=B,S$ separately.

In the Brown-Ravenhall case we choose $\psi_{j,l,m,n}^B$ such that its
Fourier transform is
$$
\hat \psi_{j,l,m,n}^B(\bp) = p^{-1} f^B_{j,l,n}(p) \Omega_{j,l,m}(\omega_\bp),
$$
where $f^B_{j,l,n}$ is the $n$-the eigenfunction of the operator
$V_c\, b_{j,l}(Z/c)\, V_c^*$ in $L^2(\rz_+)$. Here the unitary scaling
operator $V_c$ is defined by $(V_c f)(p) := c^{-1/2} f(p/c)$ and we
recall that the operator $b_{j,l}(\kappa)$ was defined in Subsection
\ref{sec:decompbr}. The operators $V_c\, b_{j,l}(Z/c)\, V_c^*$ appear
as the angular momentum reductions of $B_c[Z |\bx|^{-1}]$. Indeed, by
\eqref{2.7a} and scaling one has
\begin{equation*} 
  \langle \psi, B_c[Z |\bx|^{-1}] \, \psi\rangle 
  = c^2 \sum_{(j,l,m)\in
  \mathcal{I}} \langle \hat\psi_{j,m,l}, V_c\, b_{j,l}(Z/c)\, V_c^* \,
  \hat\psi_{j,m,l} \rangle .
\end{equation*}

In the Schr\"odinger case we choose
$$
\psi_{j,l,m,n}^S(\bx) = r^{-1} f^S_{l,n}(r) \Omega_{j,l,m}(\omega_\bx),
$$
where $f^S_{l,n}$ is the $n$-th eigenfunction of
$-\frac12\frac{\rd^2}{\rd r^2}+\frac{l(l+1)}{2 r^2}-\frac{Z}{r}$ in
$L^2(\rz_+)$ with Dirichlet boundary conditions.

\subsection{High angular momenta\label{app:a2}}

For large angular momenta, the electrons are sufficiently far from the
center moving -- classically speaking -- slowly. This motivates to
pick non-relativistic orbitals in both in the relativistic and
non-relativistic case. Moreover, for large quantum numbers the
correspondence principle would predict quasi-classical behavior (in the
quantum sense) as well. This motivates the following choice which we
take -- with slight modifications -- from \cite{SiedentopWeikard1987O}:
\begin{equation}
  \label{eq:d-groesser}
  d_> := \sum_{l\geq L} d_l, \ \
  d_l := \sum_{j=l\pm1/2}\sum_{m = -j}^{j}\sum_{n \in \zz}  w_{n,l}
  |\varphi_{n,l}\Omega_{j,l,m}\rangle\langle \varphi_{n,l}\Omega_{j,l,m}|.
\end{equation}
We repeat at this point the construction of the Macke orbitals $
\varphi_{n,l} $ and their weights $ w_{n,l} $. We will also present a
new estimate not directly given in that paper.

The semi-classical mean-field in which the electrons move is the
Thomas-Fermi potential $  \phi_\mathrm{TF} $ (see \eqref{eq:tf}).
According to
Hellmann \cite{Hellmann1936} the semi-classical electron density for
fixed angular momentum is
\begin{equation}
  \label{eq:hellmann}
  \sigma_l^H(r) :=   {2(2l+1)\over\pi} \sqrt{2 \left[
  n_Z\phi_\mathrm{TF}(r)-{(l+\eh)^2\over2 r^2}\right]_+}\,,
\end{equation}
where we added the factor $n_Z=(1-aZ^{-1/2})^{2/3}$ for normalization
purposes with some fixed positive $a$ and where we replaced the
self-generated field of the sum of the radial densities $\sigma_l$ by
the Thomas-Fermi potential.  We will write $\rho_l^H$
for the functions $\sigma_l^H$ when $a=0$, i.e., no normalization
factor occurs. In passing we note that the densities $\rho_l^H$ are
the minimizers of the Hellmann functional with external potential
given by the Thomas-Fermi density and no other interaction between the
electrons (see \cite{SiedentopWeikard1986}).

The functions $\sigma_l^H$ vanish for large $l$ and we define
$$
k' := \min\{ l\in\nz \,|\, \sigma_l^H\equiv 0 \}.
$$
By scaling, $ k' $ is of the order $Z^{1/3}$. Moreover, since the function $r\mapsto\phi_\mathrm{TF}(r) r^2$ has exactly one maximum, the support of $\sigma_l^H $ is an interval $[r_1(l),r_2(l)]$.

We cannot use the density $ \sigma_l^H $ directly in defining
semi-classical orbitals, since the derivative of its square root
is not square integrable. Thus we pick two points,
\begin{equation}
  \label{eq:x}
  x_1(l):=r_1(l)+ T(l+\eh)Z^{-1},\ x_2(l):=r_2(l)-SZ^{-2/3}
\end{equation}
for some positive $S$ and $T\in(0,4)$, and set 
\begin{equation}
  \label{eq:dichte}
  \rho_l(r):=  \begin{cases}
    2(2l+1) \alpha^2 r^{2l+2}, & r\in[0,x_1(l)],\\
    \sigma_l^H(r),             & r\in[x_1(l),x_2(l)],\\
    2(2l+1)\beta^2 \exp(-2^{3/2} Z^{2/3} r), & r\in[x_2(l),\infty).
  \end{cases}
\end{equation} 
The constants $\alpha$ and $\beta$ are chosen such that $\rho_l$ is
continuous. We suppress their dependence on $l$ in the notation.

Next, we define for $l<k'$  and $n\in\zz$ the Macke orbitals 
\begin{equation}
  \label{eq:macke}
  \varphi_{n,l}(r):={\sqrt{\zeta'_{l}(r)}\over r}e^{\ri \pi
  k_{n,l}\zeta_l(r)}
\end{equation}
where $\zeta_l:[0,\infty)\rightarrow[0,1)$ is the Macke
transform
\begin{equation}
  \label{eq:macketrafo}
  \zeta_l(r):={\int_0^r\rho_l(t)\rd t \over \int_0^\infty\rho_l(t)\rd t}.
\end{equation}
For $l\geq k'$ we set $\varphi_{n,l}:\equiv 0$.
The integral 
$$
N_{j,l,m}:= \frac1{2(2l+1)} \int_0^\infty\rho_l(r)\rd r,
$$
which is independent of $ j $ and $ m $, will represent the number of 
electrons in the angular momentum channel $ (j,l,m) $.
Moreover, we set $\varepsilon_l:= N_{j,l,m}-[N_{j,l,m}]$.
If $[N_{j,l,m}]$ is odd, we pick $k_{n,l}=2n$, otherwise
$k_{n,l}=2n-1$. The weights are chosen as
\begin{equation}
  \label{eq:gewichte}
  w_{n,l}:= 
  \begin{cases}
    1 & |k_{n,l}|\leq [N_{j,l,m}]-1\\
    \varepsilon_l/2 & |k_{n,l}|=[N_{j,l,m}]+1\\
    0 & \mbox{otherwise}
  \end{cases}
\end{equation}
which guarantees that $ \sum_{n \in \zz} w_{n,l} = N_{j,l,m} $.

  Strictly speaking, our trial density matrix differs from the one
  used in \cite{SiedentopWeikard1987O}, since we label the orbitals by
  the modulus of total angular momentum, by the third component of
  total angular momentum, and by the orbital angular momentum. This,
  however, is merely a minor rearrangement of terms. 

  We also adapt to atomic units used in this paper which changes the
  value of the Thomas-Fermi constant and gives a factor $1/2$ in front
  of all three kinetic energy terms in the Hellmann-Weizs\"acker functional. 

%%%%%%%%%%%%%%%%%%%%%%%%%%%%%%%%%%%%%%%%%%%%%%%%%%%%%%%%%%%%%%%%%%%%%%%%%%%%%

\subsection{Energy estimates for high angular momenta\label{app:b}}

For the convenience of the reader, we gather from
\cite{SiedentopWeikard1987O} (based on the construction in
\cite{Siedentop1981}) two estimates on the order of the average
kinetic and potential energy of the Schr\"odinger operator associated
with the semi-classical density matrix $ d_> $,
\begin{equation}
  \tr(\bp^2 d_>)  = \cO(Z^{7/3}) , \qquad
  \tr(|\bx|^{-1} d_>)  = \cO(Z^{4/3}) .\label{remindersw}
\end{equation}

We also need a more detailed estimate on the
kinetic energy.
\begin{lemma}
  \label{prop:kinetischeenergie}
  Let $L=[Z^{1/12}]$. Then for large $Z$,
  \begin{equation}
    \label{eq:kin}
    \sum_{l=L}^\infty l^{-2}\tr(\bp^2 d_l) = \cO(Z^2/L).
  \end{equation}
\end{lemma}
\begin{proof}
The definition of $ d_l $ implies (cf. \cite[(2.3)]{SiedentopWeikard1987O}) that
for angular momenta $ l < k'$ one has
\begin{equation}
    \label{eq:orbital}
     \tr(\bp^2d_l)  = \int_0^\infty \left[{\sqrt{\rho_l}\ }^{\prime2} +
    {\alpha_l\over3} \rho_l^3 +{l(l+1)\over r^2} \rho_l \right] \rd r +F_l
\end{equation}
where we set
\begin{equation*}
  F_l:= {\alpha_l\over3} \left({-1+6\varepsilon_l-3\varepsilon_l^2
       \over
       N^2_{j,l,m}}+{2\varepsilon_l^3-6\varepsilon_l^2+4\varepsilon_l
       \over N^3_{j,l,m}}\right) \int_0^\infty \rho_l^3\rd r,
  \qquad \alpha_l :=\frac{\pi^2}{4(2l+1)^2},
\end{equation*}
and emphasize that $\alpha_l$ should not be confused with $\alpha$
from \eqref{eq:dichte}.  According to \cite[Proposition 3.6]{SiedentopWeikard1987O}
\begin{equation*}
  \label{sumfl}
  \sum_{l=L}^\infty l^{-2} \, F_l \leq \sum_{l=L}^\infty F_l \leq \const Z^{5/3}
\end{equation*}
where $L=[Z^{1/12}]$.  The first term on the right-hand side of
\eqref{eq:orbital} is estimated according to
\begin{align}\label{eq:mainE3}
    & \int_0^\infty \left[{\sqrt{\rho_l}\ }^{\prime2} +
    {\alpha_l\over3} \rho_l^3 +{l(l+1)\over r^2} \rho_l \right] \rd r \notag \\
    & \leq \int_0^\infty \left[{\alpha_l\over3} \rho_l^H(r)^3 + {(l+\eh)^2\over
    r^2}\rho_l^H(r)\right] \rd r + G_l + H_l+I_l.
    \end{align}
   with
   \begin{align*}
  G_l&:=\int_0^{\x1}\left[ {\sqrt{\rho_l}\ }^{\prime2}+{\alpha_l \over3}
  \rho_l^3+{\betal^2 \over r^2}\rho_l \right] \rd r \leq \const Z^2
  \betal^{-3/2},\\
  H_l&:=\int_{\x2}^\infty\left[{\sqrt{\rho_l}\ }^{\prime2}+{\alpha_l
    \over3}\rho_l^3+{\betal^2\over r^2}\rho_l\right]\rd r \leq\const Z^{7/6}
    \betal 
\end{align*}
where the inequalities were obtained 
by integration as in \cite[(3.4)]{SiedentopWeikard1987O}.
Inequality \cite[(3.9)]{SiedentopWeikard1987O} for the gradient term
in the middle region reads
  \begin{multline*}
    I_l:=\int_{\x1}^{\x2}{\sqrt{\rho_l}\ }^{\prime2}\rd r
    \leq \const \betal \\
    	\times \left[ Z^2 \betal^{-3} + Z + Z^2
    \betal^{-5/2}+{Z^{5/3}\betal^{-1/2} \over \min \{ l+\eh,Z^{1/4} \}}
    \right].
  \end{multline*}
This implies
\begin{align*}
  \sum_{l=L}^\infty l^{-2} \, G_l &\leq \const Z^2 \sum_{l=L}^\infty
  l^{-7/2} \leq \const Z^2 L^{-5/2},\\ \sum_{l=L}^\infty l^{-2} \, H_l
  &\leq \sum_{l=L}^{k'} l^{-2} \,
  H_l \leq \const Z^{7/6} \log k' \leq \const Z^{7/6} \log Z \\
  \sum_{l=L}^\infty l^{-2} \, I_l &\leq \const \left[ Z^2 L^{-3} + Z
    \log Z + Z^2 L^{-5/2} + Z^{5/3} L^{-3/2} \right] \leq \const
  Z^{43/24} .
\end{align*}
It thus remains to estimate the sum of the first terms on the
right-hand side of \eqref{eq:mainE3}.  We begin with the first
summand,
\begin{align*}
  \sum_{l=L}^\infty \frac{1}{l^2} \int_0^\infty
  {\alpha_l\over3}\rho_l^H(r)^3 \rd r & \leq \const \sum_{l=L}^\infty
  \frac{1}{l}\int_0^\infty (Z/r-
  l^2/r^2)_+^{3/2}\rd r\\
  & =\const Z^2 \sum_{l\geq L} \frac{1}{l^2}\int_0^\infty r^{-3/2} (1-
  r^{-1})_+^{3/2} \rd r =\cO(Z^2/L)
\end{align*}
where we used that the Thomas-Fermi potential is bounded from above by~$Z/r$. 
This leaves the second summand,
\begin{align*}
  & \sum_{l=L}^\infty \frac{1}{l^2} \int_0^\infty{\betal^2\over
    r^2}\rho_l^H(r) \rd r
  \leq \const \sum_{l=L}^\infty l \int_0^\infty r^{-2} (Z/r - l^2/r^2)^{1/2}_+ \rd r \\
  =\ & \const Z^2 \sum_{l=L}^\infty \frac{1}{l^2} \int_0^\infty
  r^{-5/2} (1- r^{-1})^{1/2}_+ \rd r= \cO(Z^2/L) ,
\end{align*}
which completes the proof of Lemma \ref{prop:kinetischeenergie}.
\end{proof}

\end{document}